\documentclass[]{article}

\usepackage{ucs}
\usepackage[utf8x]{inputenc}
\usepackage{amsfonts,amssymb,mathtools,amsthm}
\usepackage{xspace}
\usepackage{nicefrac}
\usepackage{graphicx}
\usepackage{pifont}
\usepackage{cancel}
\usepackage{tikz}
\usepackage{setspace}
\usepackage{stackrel}
\usepackage{xifthen}
\usepackage{stmaryrd}
\usepackage{multirow}
\usepackage[normalem]{ulem}
\usepackage{float}
\usepackage{paralist}
\usepackage[left=2.25cm,right=2.25cm,top=3cm,bottom=3cm]{geometry}
\usepackage{hyperref}

\usepackage{commands}

\title{Matching in the Pi-Calculus (Technical Report)}
\author{Kirstin Peters \qquad\qquad Tsvetelina Yonova-Karbe \qquad\qquad Uwe Nestmann\vspace{0.5em}\\
TU Berlin, Germany}

\begin{document}

\maketitle


\begin{abstract}
	We study whether, in the \piCal, the {match prefix}|{a} conditional operator testing two names for (syntactic) {equality}|{is} expressible via the other operators.
	Previously, Carbone and Maffeis proved that matching is \emph{not} expressible this way under rather strong requirements (preservation and reflection of observables).
	Later on, Gorla developed a by now widely-tested set of criteria for encodings that allows much more freedom (\eg instead of direct translations of observables it allows comparison of calculi with respect to reachability of successful states).
    In this paper, we offer a considerably stronger separation result on the non-expressibility of matching using only Gorla's relaxed requirements.
    This report extends \cite{pyn14} and provides the missing proofs.
\end{abstract}


\section{Introduction}

In process calculi matching is a simple mechanism to trigger a process if two names are syntactically equal. The match prefix $ \match{a}{b}P $ in the $\pi$-calculus works as a conditional guard. If the names $ a $ and $ b $ are identical the process behaves as $ P $. Otherwise, the term cannot reduce further.

\paragraph{Motivation.} The principle of matching two names in order to reduce a term is also present in another form in any calculus with channel-based synchronisation, like CCS or the $\pi$-calculus.
The rule for communication demands identical (i.e.~matching) input/output channel names to be used by parallel processes. For example, the term $ \overline{a} \mid a.P $ may communicate on $ a $, but the term $ \overline{a} \mid b.P $ cannot communicate at all. Thus, the \piCal already contains a ``distributed'' form of the match prefix.\footnote{Of course, this observation extends to Linda-like tuple-based communication, and even to Actor-like message routing according to the matching object identity.} However, it is also an ``unprotected'' and therefore non-deterministic form of matching, as $ \overline{a} \mid a.P \mid a.Q $ allows for two different communications. This raises the natural question whether the match prefix can be encoded using the other operations of the calculus, or whether it is a basic construct.
Here we show that communication is indeed the \piCal construct that is closest to the match prefix.
Accordingly an encoding of the match prefix would need to translate the prefix into a (set of) communication step(s) on links that result from the translation of the match variables. These links have to be free|to allow for a guarding input to receive a value for a match variable|but they also have to be bound|to avoid unintended interactions between parallel match encodings. This kind of binding cannot be simulated by a \piCal operator different from the match prefix.
Thus the match prefix is a basic construct of the \piCal and cannot be encoded.
Note that, as shown by the use of the match prefix \eg in \cite{milnerParrowWalker92} for a sound axiomatisation of late congruence, in \cite{sangiorgi96} for a complete axiomatisation of open equivalence, {or}|{more recently}|{in} \cite{giunti13} for a session pi-calculus, the match prefix is regarded as useful, \ie it allows for applications that without the match prefix are not possible or more complicated to achieve.
Thus a better understanding of the nature of the match prefix contributes to current research.

\paragraph{Quality criteria.} Of course, we are not interested in trivial or meaningless encodings. Instead we consider only those encodings that ensure that the original term and its encoding show to some extent the same abstract behaviour. To analyse the quality of encodings and to rule out trivial or meaningless encodings, they are evaluated \wrt a set of quality criteria.
Note that stricter criteria that rule out more encoding attempts strengthen an \emph{encodability result}, \ie the proof of the existence of an encoding between two languages that respects the criteria. A stronger encodability result reveals a closer connection between the considered languages.
In contrast weaker criteria strengthen a \emph{separation result}, \ie the proof of the non-existence of an encoding between two languages \wrt the criteria. A stronger separation result illuminates a conceptional difference between two languages, \ie some kind of behaviour of the source language that cannot be simulated by the target language.
Unfortunately there is no consensus about what properties make an encoding ``good'' or ``good enough'' to compare two languages (compare \eg~\cite{parrow08}). Instead we find separation results as well as encodability results with respect to very different conditions, which naturally leads to incomparable results.
Among these conditions, a widely used criterion is \emph{full abstraction}, \ie the preservation and reflection of equivalences associated to the two compared languages. There are lots of different equivalences in the range of \piCal variants. Since full abstraction depends, by definition, strongly on the chosen equivalences, a variation in the respective choice may change an encodability result into a separation result, or vice versa \cite{gorlaNestmann}. Unfortunately, there is neither a common agreement about what kinds of equivalence are well suited for language comparison{|}again, the results are often incomparable.
To overcome these problems, and to form a more robust and uniform approach for language comparison, Gorla \cite{gorla} identifies five criteria as being well suited for separation as well as encodability results. By now these criteria are widely-tested (see \eg~\cite{gorla10}). Here, we rely on these criteria to measure the quality of encodings between variants of the \piCal.
Compositionality and name invariance stipulate structural conditions on a valid encoding. Operational correspondence requires that a valid encoding preserves and reflects the executions of a source term. Divergence reflection states that a valid encoding shall not exhibit divergent behaviour, unless it was already present in the source term. Finally, success sensitiveness requires that a source term and its encoding have exactly the same potential to reach a successful state.

\paragraph{Previous Results.} The question about the encodability of the match prefix is not a new one. In \cite{vig} Philips and Vigliotti proposed an encoding within the mobile ambient calculus (\cite{cardelliGordon00}). The \piCal as target language was considered by Carbone and Maffeis. They proved in \cite{carbone} that there exists no encoding of the \piCal into the \piCal (with only guarded choice and) without the match prefix. However the quality criteria used in \cite{carbone} are more restrictive than the criteria here. In particular they assume that visible communication links, \ie observables, are preserved and reflected by the encoding, \ie a source term and its encoding must have the same observables. This criterion is very limiting (\ie strict) even for an encoding between two variants of the same calculus. Thus, by using weaker quality criteria, we strengthen the separation result presented in \cite{carbone}. Note that we use \cite{carbone} as a base and starting point for our result. We discuss the differences to the proofs of \cite{carbone} in Section~\ref{sec:comparison}.
In the same paper Carbone and Maffeis show that the match prefix can be encoded by polyadic synchronisation.
Another positive result for a variant of the \piCal is presented by Vivas in \cite{vivas}. There, a modified version of the \piCal with a new operator, called blocking, is used to encode the match prefix.
In \cite{bodeiDeganoPriami05} the input prefix of the \piCal is replaced by a selective input that allows for communication only if the transmitted value is contained in a set of names specified in the selective input prefix. Accordingly selective input can be used as conditional guard, which can replace the match prefix. We discuss these encoding approaches and how they are related to our separation result in Section~\ref{sec:encodeMatchInOtherCalculi}.

\paragraph{Overview.} We start with an introduction of the considered variants of the \piCal in \S\ref{sec:processCalculi}. Then \S\ref{sec:quality} introduces the framework of \cite{gorla} to measure the quality of an encoding. Our separation result is presented in \S\ref{sec:encodeMatch}. In \S\ref{sec:discussion} we discuss the relation of our result to related work. We conclude with \S\ref{sec:conclusions}.


\section{The Pi-Calculus}
\label{sec:processCalculi}

Within this paper we compare two different variants of the \piCal|the full \piCal with (free choice and) the match prefix (\piT) and its variant without the match prefix (\piNM)|as they are described \eg in \cite{milnerParrowWalker92,Milner1999}.

Let $ \mc $ denote a countably infinite set of names and $ \overline{\mc} $ the set of co-names, \ie $ \overline{\mc} = \Set{ \overline{n} \mid n \in \mc } $. We use lower case letters $ a, a', a_1, \ldots, x, y, \ldots $ to range over names.
Moreover let $ \mc^k $ denote the set of vectors of names of length $ k $. Let $ \mc^* $ be the set of finite vectors of names. And let $ \proj{\tilde{x}}{i} = x_i $ whenever $ \tilde{x} = x_1, \ldots, x_n $ and $ 1 \leq i \leq n $.
For simplicity we adapt some set notations to deal with vectors of names, \eg $ \length{\tilde{x}} $ is the length of the vector $ \tilde{x} $, $ a \in \tilde{x} $ holds if the name $ a $ occurs in the vector $ \tilde{x} $, and $ \tilde{x} \cap \tilde{y} = \emptyset $ holds if the vectors $ \tilde{x} $ and $ \tilde{y} $ do not share a name.

\begin{definition}[Syntax]
	The set of process terms of the \emph{full \piCal}, denoted by $ \procPi $, is given by
	\begin{align*}
		P \deffTerms & \nullTerm \sep \piInput{x}{z}{P} \sep \piOutput{x}{y}{P} \sep \tau.P \sep \match{a}{b}P \sep\\
		& P_1 + P_2 \sep P_1 \mid P_2 \sep \Res{z}{P} \sep !P \sep \success
	\end{align*}
	
	\noindent
	where $ a, b, x, y, z \in \mc $.
	The processes of its subcalculus \piNM, denoted by $ \procPiNoMatch $, are given by the same grammar without the match prefix $ \match{a}{b}P $.
\end{definition}

\noindent
The term $ \nullTerm $ denotes an inactive process, \ie a process that can do nothing.
The remaining operators of the first line define guards.
The input prefix $ \piIn{x}{z} $ is used to describe the ability of receiving the value $ z $ over link $ x $ and, analogously, the output prefix $ \piOut{x}{y} $ describes the ability to send a value $ y $ over link $ x $. We call $ x $ the subject of an action prefix $ \piIn{x}{z} $ or $ \piOut{x}{y} $ and $ y $ or $ z $ its object. After receiving or sending a value the processes $ \piInput{x}{z}{P} $ and $ \piOutput{x}{y}{P} $ continue as $ P $.
The \emph{unobservable prefix} $ \tau $ defines an internal action. The process $ \tau.P $ can evolve invisibly to $ P $.
The \emph{match prefix} $ \match{a}{b} $ works as a conditional guard. It can be removed iff $ a $ and $ b $ are equal. If $ a = b $ then $ \match{a}{b}P $ continues as $ P $.
Sometimes we denote the $ a $ and $ b $ in $ \match{a}{b}P $ as \emph{match variables}.
The \emph{sum} $ P_1 + P_2 $ defines a \emph{choice}, \ie the process which either behaves as $ P_1 $ or $ P_2 $. Note that we consider two variants of the \piCal with finite \emph{free choice}, \ie choice terms are not limited to guarded summands.
$ P_1 \mid P_2 $ defines \emph{parallel composition}, \ie the process in which $ P_1 $ and $ P_2 $ may proceed independently, possibly interacting using shared links.
\emph{Restriction} $ \Res{z}{P} $ restricts the scope of the name $ z $ to the definition of $ P $. We use $ \Res{\tilde{x}}{P} $ to abbreviate $ \Res*{x_1}{ \Res*{x_2}{ \ldots (\Res{x_n}{P}) \ldots}} $ for some vector $ \tilde{x} = x_1, \ldots, x_n \in \mc^* $ of names.
$ !P $ denotes \emph{replication} and can be thought of as an infinite parallel composition of $ P $.
The term $ \success $ denotes \emph{success} (or \emph{successful termination}). It is introduced in order to compare the abstract behaviour of terms in different process calculi as described in Section~\ref{sec:quality}.

We use $ P, P', P_1, \ldots, Q, R, \ldots $ to range over processes. Let $ \freeNames{P} $, $ \boundNames{P} $, and $ \names{P} $ denote the sets of \emph{free names} in $ P $, \emph{bound names} in $ P $, and all \emph{names} occurring in $ P $, respectively. Their definitions are completely standard, \ie names are bound by restriction and as parameter of input and $ \names{P} = \freeNames{P} \cup \boundNames{P} $ for all $ P $.

We use $ \sigma $, $ \sigma' $, $ \sigma_1 $, \ldots to range over substitutions. A substitution is a finite mapping from names to names defined by a set $ \Set{ \subs{y_1}{x_1}, \ldots, \subs{y_n}{x_n} } $ of renamings, where the $ x_1, \ldots, x_n $ are pairwise distinct.
The term $ \Set{ \subs{y_1}{x_1}, \ldots, \subs{y_n}{x_n} }\left( P \right) $ is defined as the result of simultaneously replacing all free occurrences of $ x_i $ by $ y_i $ for $ i \in \Set{ 1, \ldots, n } $, possibly applying alpha-conversion to avoid capture or name clashes. For all names $ \mc \setminus \Set{ x_1, \ldots, x_n } $ the substitution behaves as the identity mapping, \ie as empty substitution.
Let $ \id \deff \emptyset $ denote the identity mapping.
We call the names that are modified by a substitution $ \sigma $ its \emph{domain}, \ie $ \dom{\sigma} \deff \Set{ x \mid \subs{y}{x} \in \sigma \wedge x \neq y } $. Similarly $ \codom{\sigma} \deff \Set{ y \mid \subs{y}{x} \in \sigma \wedge x \neq y } $ and $ \names{\sigma} \deff \dom{\sigma} \cup \codom{\sigma} $.
We naturally extend substitutions to co-names, \ie for all $ \sigma : \mc \to \mc $ and all $ n \in \mc $, $ \sigma\!\left( \overline{n} \right) = \overline{\sigma\!\left( n \right)} $.

\begin{figure}[t]
	\begin{align*}
		& \hspace{5.7em} \begin{array}{r@{\quad}c@{\quad}ll}
			P & \equiv & Q & \quad
				\begin{aligned}[t]
					& \text{ if } Q \text{ can be obtained from } P \text{ by renaming one or more of the}\\
					& \text{ bound names in P, silently avoiding name clashes }
				\end{aligned}
		\end{array}\\
		& \begin{array}{r@{\quad}c@{\quad}lcr@{\quad}c@{\quad}l}
			\match{a}{a}P & \equiv & P & \hspace{3em} & !P & \equiv & P \mid \; !P \vspace{0.5em}\\
			P_1 + \left( P_2 + P_3 \right) & \equiv & \left( P_1 + P_2 \right) + P_3 & & P_1 \mid \left( P_2 \mid P_3 \right) & \equiv & \left( P_1 \mid P_2 \right) \mid P_3\\
			P_1 + P_2 & \equiv & P_2 + P_1 & & P_1 \mid P_2 & \equiv & P_2 \mid P_1\\
			P + 0 & \equiv & P & & P \mid 0 & \equiv & P \vspace{0.5em}\\
			\Res{z}{\Res{w}{P}} & \equiv & \Res{w}{\Res{z}{P}}\\
			\Res{z}{\nullTerm} & \equiv & \nullTerm\\
			\Res*{z}{P_1 \mid P_2} & \equiv & P_1 \mid \Res{z}{P_2} & & \text{if } z \notin \freeNames{P_1}
		\end{array}
	\end{align*}
	\caption{Structural Congruence.}
	\label{fig:structuralCongruence}
\end{figure}

As suggested in \cite{gorla} we use a \emph{reduction semantics} to reason about the behaviour of \piT and \piNM.
The \emph{reduction semantics} of \piT and \piNM are jointly given by the transition rules in Figure~\ref{fig:reductionSemantics}, where \emph{structural congruence}, denoted by $ \equiv $, is given by the rules in Figure \ref{fig:structuralCongruence}. As usual, we use $ \equiv_{\alpha} $ if we refer to alpha-conversion (the first rule of Figure~\ref{fig:structuralCongruence}) only.
Note that Figure~\ref{fig:reductionSemantics} defines not only reduction rules ($ \step $) but also some rules for labeled steps ($ \labeledStep{\alpha} $, $ \labeledStep{\piIn{x}{y}} $, and $ \labeledStep{\piOut{x}{y}} $). They allow us to deal with arbitrary nestings of choice and parallel compositions. However, our separation result is based on the reduction semantics of the respective calculi. The reduction semantics in Figure~\ref{fig:reductionSemantics} coincide to the reduction semantics given be the rules:
\begin{center}
	$ \tau.P \step P \hspace*{2em} \piOutput{x}{y}{P} + P' \mid \piInput{x}{z}{Q} + Q' \step P \mid \Set{ \subst{y}{z} }Q \hspace*{2em} \dfrac{P \step P'}{P + Q \step P'} $\vspace{0.75em}\\
	$ \dfrac{P \step P'}{P \mid Q \step P' \mid Q} \hspace*{2em} \dfrac{P \step P'}{\Res{n}{P} \step \Res{n}{P'}} \hspace*{2em} \dfrac{P \equiv Q \quad \quad Q \step Q' \quad Q' \equiv P'}{P \step P'} $\vspace{0.75em}
\end{center}

\begin{figure}[t]
	\begin{align*}
		\begin{array}{l@{\quad}c@{\hspace{3em}}l@{\quad}c}
			\textbf{Input} & \dfrac{}{\piInput{x}{y}{P} + Q \labeledStep{\piIn{x}{y}} P} & \textbf{Output} & \dfrac{}{\piOutput{x}{y}{P} + Q \labeledStep{\piOut{x}{y}} P}\\
			\\
			\textbf{Com} & \dfrac{P \labeledStep{\piIn{x}{y}} P' \quad Q \labeledStep{\piOut{x}{y}} Q'}{P \mid Q \step P' \mid Q'} & \textbf{Tau} & \dfrac{}{\tau.P + Q \step P}\\
			\\
			\textbf{Sum-l} & \dfrac{P \labeledStep{\alpha} P'}{P + Q \labeledStep{\alpha} P'} & \textbf{Sum} & \dfrac{P \step P'}{P + Q \step P'}\\
			\\
			\textbf{Par-l} & \dfrac{P \labeledStep{\alpha} P'}{P \mid Q \labeledStep{\alpha} P' \mid Q} & \textbf{Par} & \dfrac{P \step P'}{P \mid Q \step P' \mid Q}\\
			\\
			& & \textbf{Res} & \dfrac{P \step P'}{\Res{z}{P} \step \Res{z}{P'}}\\
			\\
			\textbf{Cong-l} & \dfrac{P \equiv P' \quad P' \labeledStep{\alpha} Q' \quad Q' \equiv Q}{P \labeledStep{\alpha} Q} & \textbf{Cong} & \dfrac{P \equiv P' \quad P' \step Q' \quad Q' \equiv Q}{P \step Q}
		\end{array}
	\end{align*}
	\caption{Reduction Semantics.}
	\label{fig:reductionSemantics}
\end{figure}

Note that the structural congruence rule $ \match{a}{a}P \equiv P $ can be applied only in the full \piCal. It is this structural congruence rule (in combination with the last transition rule) that defines the semantics of the match prefix. However we can similarly define the semantics of the match prefix with the reduction rule $ \match{a}{a}P \step P $ without any influences on our results.
A reduction step $ P \step P' $ then denotes either a communication between an input and output on the same link or an internal step.
Let $ P \step $ (and $ P \noStep $) denote the existence (and non-existence) of a step from $ P $, \ie there is (no) $ P' $ such that $ P \step P' $. Moreover, let $ \steps $ be the reflexive and transitive closure of $ \step $. We write $ P \step^{\omega} $ if $ P $ can perform an infinite sequence of reduction steps. A sequence of reduction steps starting in a term $ P $ is called an \emph{execution} of $ P $. An execution is either finite, as $ P_0 \step P_1 \step \ldots \step P_n $, or infinite. A finite execution $ P_0 \steps P_n $ is \emph{maximal} if it cannot be further extended, \ie if $ P_n \noStep $, otherwise it is \emph{partial}.

Analysing the reduction rules we observe that substitutions can enable new communication steps|by unifying the links of input and output guards|but they never disable steps.

\begin{obs}
	\label{obs:subsSteps}
	Let $ P, P' $ be processes, \ie $ P, P' \in \procPi $ or $ P, P' \in \procPiNoMatch $, and $ \sigma : \mc \to \mc $ be a substitution.
	Then:
	\begin{align*}
		P \steps P' \impl \sigma\!\left( P \right) \steps \sigma\!\left( P' \right)
	\end{align*}
\end{obs}

Traditionally a process term is considered as successful if it has an unguarded occurrence of success (see \eg \cite{gorla}). This is usually formalised as $ \exists P' \pkt P \equiv \success \mid P' $. Because of free choice, we have to adapt the usual definitions of the reachability of success to deal with arbitrary nestings of choice and parallel composition. To do so we recursively define the notion of unguarded subterms.

\begin{definition}[Unguarded Subterms]
	\label{def:unguardedSubterms}
	Let $ P \in \procPi $ or $ P \in \procPiNoMatch $. The \emph{set of unguarded subterms of} $ P $, denoted by $ \ungSub{P} $, is recursively defined as:
	\vspace*{-0.6em}
	\begin{align*}
  		\begin{cases}
  			\Set{ P } \cup \ungSub{Q} & \text{, if } P = \match{a}{a}Q\\
  			\Set{ P } \cup \ungSub{Q_1} \cup \ungSub{Q_2} & \text{, if } P = Q_1 + Q_2 \vee P = Q_1 \mid Q_2\\
  			\Set{ P } \cup \ungSub{Q} & \text{, if } P = \Res{z}{Q} \vee P = {!}Q\\
  			\Set{ P } & \text{, otherwise}
  		\end{cases}
	\end{align*}
	\vspace*{-1.1em}
\end{definition}

\noindent
Note that the sets of unguarded subterms can differ for structural congruent terms. Consider for example $ \ungSub{\Res{z}{\piOutput{z}{z}{\nullTerm}}} = \Set{ \Res{z}{\piOutput{z}{z}{\nullTerm}}, \piOutput{z}{z}{\nullTerm} } $ but $ \ungSub{\Res{z'}{\piOutput{z'}{z'}{\nullTerm}}} = \Set{ \Res{z'}{\piOutput{z'}{z'}{\nullTerm}}, \piOutput{z'}{z'}{\nullTerm} } $
or $ \ungSub{\success} = \Set{ \success } $ but $ \ungSub{\success + \nullTerm} = \Set{ \success + \nullTerm, \success, \nullTerm } $.
Similarly, injective substitutions do not distribute over unguarded subterms.
For example $ \Set{ \subst{y}{x} }\!\left( \ungSub{\Res{x}{\piOutput{x}{x}{\nullTerm}}} \right) = \Set{ \Res{x}{\piOutput{x}{x}{\nullTerm}}, \piOutput{y}{y}{\nullTerm} } $ but $ \Set{ \subst{y}{x} }\!\left( \ungSub{\Res{z}{\piOutput{z}{z}{\nullTerm}}} \right) = \Set{ \Res{z}{\piOutput{z}{z}{\nullTerm}}, \piOutput{z}{z}{\nullTerm} } $.
Moreover note that if $ P' $ is an unguarded subterm of $ P $ then also all unguarded subterms of $ P' $ are unguarded subterms of $ P $.

From Definition~\ref{def:unguardedSubterms} we conclude that if $ P' $ is an unguarded subterm of $ P $ then also all unguarded subterms of $ P' $ are unguarded subterms of $ P $.

\begin{lemma}
	\label{prop:subsetUngSub}
	Let $ P $ and $ P' $ be processes, \ie either $ P, P' \in \procPi $ or $ P, P' \in \procPiNoMatch $. Then $ P' \in \ungSub{P} $ implies $ \ungSub{P'} \subseteq \ungSub{P} $.
\end{lemma}

Then a term is successful if it has an unguarded occurrence of success.

\begin{definition}[Reachability of Success]
	Let $ P \in \procPi $ or $ P \in \procPiNoMatch $. Then $ P $ is \emph{successful}, denoted by $ P \hasS $, if $ \success \in \ungSub{P} $.
	\emph{$ P $ reaches success}, denoted by $ P \reachS $, if there is some $ Q $ such that $ P \steps Q $ and $ Q \hasS $.
	Moreover, we write $ P \mustReachSuccessFinite $, if $ P $ reaches success in every finite maximal execution.
	Let $ P \nHasS $ abbreviate $ \neg \left( P \hasS \right) $, $ P \nReachS $ abbreviate $ \neg \left( P \reachS \right) $, and $ P \not\mustReachSuccessFinite $ abbreviate $ \neg \left( P \mustReachSuccessFinite \right) $.
\end{definition}

\noindent
Of course, all proofs in this paper hold similarly for variants of \piT and \piNM with only guarded choice and the traditional definition of a successful term.

The first quality criterion to compare process calculi presented in Section~\ref{sec:quality} is compositionality. It induces the definition of a \piNM-context parametrised on a set of names for each operator of \piT. A \piNM-context $ \context{C}{\hole_1, \ldots, \hole_n} : (\procPiNoMatch)^n \to \procPiNoMatch $ is simply a \piNM-term with $ n $ holes. Putting some \piNM-terms $ P_1, \ldots, P_n $ in this order into the holes $ \hole_1, \ldots, \hole_n $ of the context, respectively, gives a term denoted by $ \context{C}{P_1, \ldots, P_n} $. Note that a context may bind some free names of $ P_1, \ldots, P_n $. The arity of a context is the number of its holes. We extend the definition of unguarded subterms by the equation $ \ungSub{\hole} = \Set{ \hole } $ to deal with contexts.

The standard notion of equivalence to compare terms of the \piCal is bisimulation. An introduction to bisimulations in the \piCal can be found \eg in \cite{milnerParrowWalker92} or \cite{sang}. For our separation result we require such a standard version of reduction bisimulation, denoted by $ \asymp $, on the target language, \ie on \piNM-terms.


\section{Quality of Encodings}
\label{sec:quality}

Within this paper we analyse the existence of an encoding from \piT into \piNM. To measure the quality of such an encoding, Gorla \cite{gorla} suggested five criteria well suited for language comparison. Accordingly, we consider an encoding to be ``valid'', if it satisfies Gorla's five criteria.

We call the tuple $ \mathcal{L} = \left( \mathcal{P}, \step \right) $, where $ \mathcal{P} $ is a set of language terms and $ \step $ is a reduction semantics, a \emph{language}.
An \emph{encoding} from $ \mathcal{L}_1 = \left( \mathcal{P}_1, \step_1 \right) $ into $ \mathcal{L}_2 = \left( \mathcal{P}_2, \step_2 \right) $ is then a tuple $(\enco{\cdot}, \vap, \asymp )$ such that
\begin{compactitem}
	\item $ \enco{\cdot} : \mathcal{P}_1 \rightarrow \mathcal{P}_2 $ is the translating function,
	\item $ \vap : \mc \rightarrow \mck $ is a renaming policy, where $ \vap(u) \cap \vap(v) = \emptyset$ for all $ u \neq v $,
	\item and $ \asymp $ is a behaviour equivalence on $ \mathcal{L}_2 $.
\end{compactitem}
We call $ \lang_1 $ the \emph{source language (calculus)} and $ \lang_2 $ the \emph{target language (calculus)}. Accordingly we call the elements of $ \proc_1 $ \emph{source terms} and the elements of $ \proc_2 $ \emph{target terms}.
We use $ S, S', S_1, \ldots $ ($ T, T', T_1, \ldots $) to range over source (target) terms.

The main ingredient of an encoding is of course the encoding function $ \enco{\cdot} $ that is a mapping from processes to processes.
However, sometimes it is useful to be able to reserve some names to play a special role in an encoding.
Since most process calculi have infinitely many names in their alphabet, it suffices to shift the set of names $ \Set{ x_0, x_1, \ldots } $ of the target language to the set $ \Set{ x_n, x_{n+1}, \ldots } $ to reserve $ n $ names. In order to incorporate such ``shifts'' and similar techniques, Gorla introduces a renaming policy $ \vap $, \ie mapping from names to names that specifies the translation of each name of the source language into a name or vector of names of the target language.
Additionally we assume the existence of a behavioural equivalence $ \asymp $ on the target language that is a reduction bisimulation. Its purpose is to describe the abstract behaviour of a target process, where abstract basically means with respect to the behaviour of the source term. Therefore it should abstract from ``junk'' left over by the encoding.

\cite{gorla} requires $ \vap $ to map all names to a vector of the same length since this way names are treated uniformly, \ie source names cannot be handled differently by an encoding just because the length of the vector, to that $ \vap $ maps to, is different.
Consequently all the names in the vector generated by $ \vap $ for a source term name should be pairwise different.

\begin{obs}
	\label{obs:renamingPolicyVector}
	Let \encod be a valid encoding from \piT into \piNM. Then:
	\begin{align*}
		\forall x \in \mc \pkt \forall i, j \in \Set{ 1, \ldots, \length{\vap\!\left( x \right)} } \pkt \proj{\vap\!\left( x \right)}{i} = \proj{\vap\!\left( x \right)}{j} \iff i = j
	\end{align*}
\end{obs}

The condition that $\vap(u) \cap \vap(v) = \emptyset$ for all $u \neq v$ ensures that the renaming policy does not relate unrelated source term names.

The five criteria that are proposed in \cite{gorla} are divided into two structural and three semantic criteria. The structural criteria comprise (1) \emph{compositionality} and (2) \emph{name invariance}. The semantic criteria comprise (3) \emph{operational correspondence}, (4) \emph{divergence reflection} and (5) \emph{success sensitiveness}.

Intuitively, an encoding is compositional if the translation of an operator is similar for all its parameters. To mediate between the translations of the parameters the encoding defines a unique context for each operator, whose arity is the arity of the operator. Moreover, the context can be parametrised on the free names of the corresponding source term. Note that our result is independent of this parametrisation.

\bkrit[Compositionality]
	\label{comp}
	A translation $\encod$ from $\mathcal{L}_1$ into $\mathcal{L}_2$ is \emph{compositional}, if for each $k$-ary operator $\mathtt{op}$ of $\mathcal{L}_1$ and for each subset of names $N$ there is a $k$-ary context $ \contextOP{C}{N}{\mathtt{op}}{\hole_1, \dots, \hole_k} $, such that for all $S_1, \dots, S_k$ with $\freeNames{S_1, \dots, S_k} = N$ it holds that $ \enco{\mathtt{op}(S_1, \dots, S_k)} = \contextOP{C}{N}{\mathtt{op}}{\enco{S_1}, \ldots, \enco{S_k}} $.
\ekrit

The second structural criterion of Gorla states that the encoding should not depend on specific names used in the source term. Therefore it describes how encoding functions have to deal with substitutions.

\bkrit[Name Invariance]
	\label{ninv}
	A translation $\encod$ from $\mathcal{L}_1$ into $\mathcal{L}_2$ is \emph{name invariant}, if for each $S$ and $\sigma$ it holds that 
$$\enco{\sigma(S)} \left\{\begin{array}{ll}= \sigma'\left(\enco{S}\right)&\mbox{ if }\sigma\mbox{ is injective}\\\asymp \sigma'\left(\enco{S}\right)&\mbox{ otherwise}\end{array}\right.$$ 
where $\sigma'$ is such that $ \vap(\sigma(a)) = \sigma'(\vap(a)) $ for every $a \in \mc $.
\ekrit

$ \sigma' $ can be considered as the translation of $ \sigma $. The condition $ \vap(\sigma(a)) = \sigma'(\vap(a)) $ ensures that $ \vap $ introduces no additional renamings between (parts of) translations of source term names. Of course $ \sigma' $ cannot affect reserved names, \ie for all names $ x $ in the domain of $ \sigma' $ there is a source term name $ a $ such that $ x \in \vap\!\left( a \right) $.

The first semantic criterion is operational correspondence, which consists of a soundness and a completeness condition.
\emph{Completeness} requires that every execution of a source term can be simulated by its translation, \ie the translation does not omit any executions of the source term. \emph{Soundness} requires that every execution of a target term corresponds to some execution of the corresponding source term, \ie the translation does not introduce new executions.

\bkrit[Operational Correspondence]
	\label{opco}
	A $\encod$ from $\mathcal{L}_1$ into $\mathcal{L}_2$ is \emph{operationally corresponding}, if it is\\
	\begin{tabular}{l@{\;\;}l}
		\emph{Complete:} & for all $ S \steps S' $, it holds that $ \enco{S} \steps \asymp \enco{S'} $;\\
		\emph{Sound:} & for all $ \enco{S} \steps T $, there exists an $ S' $\\
		& such that $ S \steps S' $ and $ T \steps \asymp \enco{S'} $.
	\end{tabular}
\ekrit

\noindent
Note that the Definition of operational correspondence relies on the equivalence $ \asymp $ to get rid of junk possibly left over within executions of target terms. Sometimes, we refer to the completeness criterion of operational correspondence as operational completeness and, accordingly, for the soundness criterion as operational soundness.

The next criterion deals with infinite computations. It states that an encoding should not introduce divergent executions.

\bkrit[Divergence Reflection]
	\label{divre}
	A translation $\encod$ from $\mathcal{L}_1$ into $\mathcal{L}_2$ \emph{reflects divergence}, if for every $S$ with $ \enco{S} \step^{\omega} $, it holds that $ S \step^{\omega} $.
\ekrit

The last criterion links the behaviour of source terms to the behaviour of target terms. With \cite{gorla}, we assume a \emph{success} operator $ \success $ to be part of the syntax of both the source and the target language.
Since $ \success $ cannot be further reduced, the operational semantics is left unchanged. Moreover, note that $ \names{\success} = \freeNames{\success} = \boundNames{\success} = \emptyset $, so also interplay of $\success$ with the $\equiv$-rules is smooth and does not require explicit treatment.
An encoding respects the behaviour of the source term if it and its translation answer the tests for success in exactly the same way.

\bkrit[Success Sensitiveness]
	\label{succse}
	A translation $\encod$ from $ \mathcal{L}_1 $ into $ \mathcal{L}_2 $ is \emph{success sensitive}, if for every $S$, it holds $ S \reachS $ iff $ \enco{S} \reachS $.
\ekrit

If an encoding satisfies all five criteria we call it valid.

\begin{definition}[Valid Encoding]
	An encoding from $ \lang_1 = \left( \mathcal{P}_1, \step_1 \right) $ into $ \lang_2 = \left( \mathcal{P}_2, \step_2 \right) $ is \emph{valid} if it satisfies compositionality, name invariance, operational correspondence, divergence reflection, and success sensitiveness.
\end{definition}

Success sensitiveness only links the behaviours of source terms and their literal translations but not of their continuations. To do so, Gorla relates success sensitiveness and operational correspondence by requiring that $ \asymp $ never relates two processes that differ in the possibility to reach success. More precisely $ \asymp $ \emph{respects success} if, for every $ P $ and $ Q $ with $ P \reachS $ and $ Q \nReachS $, it holds that $ P \not\asymp Q $.
By \cite{gorla} a ``good'' equivalence $ \asymp $ is often defined in the form of a barbed equivalence (as described e.g. in \cite{milnerSangiorgi92}) or can be derived directly from the reduction semantics and is often a congruence, at least with respect to parallel composition. For the separation results presented in this paper, we require only that $ \asymp $ is a success respecting reduction bisimulation, \ie for every $ T_1, T_2 \in \proc_2 $ such that $ T_1 \asymp T_2 $, $ T_1 \reachS $ iff $ T_2 \reachS $ and for all $ T_1 \steps_2 T_1' $ there exists a $ T_2' $ such that $ T_2 \steps_2 T_2' $ and $ T_1' \asymp T_2' $.

The following two Lemmata are proved in \cite{petphd}. Their proofs do not rely on on specific source or target languages and thus hold also in the present case. For simplicity and since the proofs are short, we repeat them adapted to the notions of this paper. The first shows that the bisimulation $ \asymp $ on target terms respects also the ability to reach success in all finite maximal executions.

\begin{lemma}
	\label{lem:mustSuccessRespecting}
	Let \encod be a valid encoding from \piT into \piNM.
	Let $ T_1, T_2 \in \procPiNoMatch $ such that $ T_1 \asymp T_2 $.
	Then:
	\begin{align*}
		T_1 \mustReachSuccessFinite \iff T_2 \mustReachSuccessFinite
	\end{align*}
\end{lemma}

\begin{proof}
	If $ T_1\mustReachSuccessFinite $ but not $ T_2 \mustReachSuccessFinite $ then, for all $ T_1' \in \procPiNoMatch $ with $ T_1 \steps T_1' $, we have $ T_1' \reachS $ but there exists some $ T_2' \in \procPiNoMatch $ such that $ T_2 \steps T_2' $ and $ T_2' \nReachS $.
	Since $ \asymp $ is a bisimulation, $ T_1 \asymp T_2 $ and $ T_2 \steps T_2' $ imply that there is some $ T_1'' \in \procPi $ such that $ T_1 \steps T_1'' $ and $ T_2' \asymp T_1'' $.
	Because $ \asymp $ respects success, $ T_2' \asymp T_1'' $ and $ T_2' \nReachS $ imply $ T_1'' \nReachS $.
	This violates the requirement that $ T_1 \mustReachSuccessFinite $, \ie contradicts the assumption that for all $ T_1' \in \procPiNoMatch $ with $ T_1 \steps T_1' $ we have $ T_1' \reachS $.
	We conclude that $ T_1\mustReachSuccessFinite $ iff $ T_2\mustReachSuccessFinite $.
\end{proof}

The second states that a valid encoding is also sensitive to the ability to reach success in all finite maximal executions.

\begin{lemma}
	\label{lem:mustSuccessSensitiveness}
	Let \encod be a valid encoding from \piT into \piNM.
	Let $ S \in \procPi $.
	Then:
	\begin{align*}
		S \mustReachSuccessFinite \impl \enco{S} \mustReachSuccessFinite
	\end{align*}
\end{lemma}

\begin{proof}
	Assume the contrary, \ie there is some $ S \in \procPi $ such that $ S \mustReachSuccessFinite $, but there is some $ T \in \procPiNoMatch $ such that $ \enco{S} \steps T $ and $ T \nReachS $.
	Since $ \enco{\cdot} $ is operationally sound, $ \enco{S} \steps T $ implies that there are $ S'' \in \procPi $ and $ T' \in \procPiNoMatch $ such that $ S \steps S'' $ and $ T \steps T' \asymp \enco{S''} $. Then $ T \nReachS $ and $ T \steps T' $ imply $ T' \nReachS $. Since $ \asymp $ respects success, $ T' \asymp \enco{S''} $ and $ T' \nReachS $ imply $ \enco{S''} \nReachS $. Because of success sensitiveness, then also $ S'' \nReachS $, which contradicts the assumption that $ S\mustReachSuccessFinite $.
\end{proof}


\section{The Match Prefix and the Pi-Calculus}
\label{sec:encodeMatch}

Our separation result strongly rests on the criteria compositionality and success sensitiveness.
We also make use of name invariance.
But, as we claim, name invariance is not crucial for the proof. Name invariance defines how a valid encoding has to deal with substitutions. This is used to simplify the argumentation in our proof as explained below.
The last criterion states that source and target terms are related by their ability to reach success.
If we compare \piT and \piNM we observe a difference with respect to successful terms and substitutions.
In \piT a substitution can change the state of a process from unsuccessful to successful. Consider for example the term $ \match{a}{b}\success $ and a substitution $ \sigma $ such that $ \sigma(a) = \sigma(b) $. The only occurrence of success in $ \match{a}{b}\success $ is guarded by a match prefix and thus $ \left( \match{a}{b}\success \right) \nHasS $. But $ \sigma\!\left( \match{a}{b}\success \right) = \match{\sigma(a)}{\sigma(b)}\success $ and thus $ \sigma\!\left( \match{a}{b}\success \right) \hasS $.
In \piNM, because there is no match prefix, a substitution cannot turn an unsuccessful state into a successful state.

\begin{lemma}
	\label{prop:propequiv}
	Let $ T \in \procPiNoMatch $. Then $ T \hasS \iff \forall \sigma : \mc \to \mc \pkt \sigma\!\left( T \right) \hasS $.
\end{lemma}

\begin{proof}
	Assume $ \forall \sigma : \mc \to \mc \pkt \sigma\!\left( T \right) \hasS $. Then, to show $ T \hasS $, it suffices to choose $ \sigma = \id $.
	The proof of the opposite direction is by induction over the structure of $ T $.
	\begin{description}
		\item[Base Cases:] Neither $ \nullTerm $ nor $ \success $ contain names. Hence if $ T = \nullTerm $ or $ T = \success $ then $ \sigma\!\left( T \right) = T $ for all $ \sigma : \mc \to \mc $.
		\item[Induction Hypotheses:] Let $ P_1, P_2 \in \procPiNoMatch $ such that
			\begin{align*}
				P_1 \hasS \impl \forall \sigma : \mc \to \mc \pkt \sigma\!\left( P_1 \right) \hasS \quad \text{ and } \quad P_2 \hasS \impl \forall \sigma : \mc \to \mc \pkt \sigma\!\left( P_2 \right) \hasS
			\end{align*}
		\item[Induction Steps:] Let $ x, y, z \in \mc $. We have to consider seven cases.
			\begin{description}
				\item[Case $ T = \piInput{x}{z}{P_1} $:] Then $ \success \notin \ungSub{\piInput{x}{z}{P_1}} = \Set{ \piInput{x}{z}{P_1}} $, \ie all occurrences of success in T|if there are any|are guarded, and thus $ T \nHasS $ and $ \sigma\!\left( T \right) \nHasS $ for all $ \sigma $.
				\item[Case $ T = \piOutput{x}{y}{P_1} $:] Similar to the first case.
				\item[Case $ T = \tau.P_1 $:] Similar to the first case.
				\item[Case $ T = P_1 + P_2 $:] Assume $ T \hasS $. Then $ T = P_1 + P_2 $ implies that $ P_1 \hasS $ or $ P_2 \hasS $. By the induction hypotheses then $ T \hasS $ implies either $ \forall \sigma : \mc \to \mc \pkt \sigma\!\left( P_1 \right) \hasS $ or $ \forall \sigma : \mc \to \mc \pkt \sigma\!\left( P_2 \right) \hasS $. Note that $ \sigma\!\left( T \right) = \sigma\!\left( P_1 + P_2 \right) \stackrel{\text{Def. } \sigma}{=} \sigma\!\left( P_1 \right) + \sigma\!\left( P_2 \right) $ for all substitutions $ \sigma : \mc \to \mc $. Thus $ \sigma\!\left( P_1 \right) \hasS $ implies $ \sigma\!\left( T \right) \hasS $, because $ \ungSub{\sigma\!\left( P_1 \right)} \subseteq \ungSub{\sigma\!\left( T \right)} $. Similarly $ \sigma\!\left( P_2 \right) \hasS $ implies $ \sigma\!\left( T \right) \hasS $. In both cases we conclude $ T \hasS \impl \forall \sigma : \mc \to \mc \pkt \sigma\!\left( T \right) \hasS $.
				\item[Case $ T = P_1 \mid P_2 $:] Similar to the case above.
				\item[Case $ T = \Res{z}{P_1} $:] Assume $ T \hasS $. Then $ T = \Res{z}{P_1} $ implies $ P_1 \hasS $. By the first induction hypothesis then $ T \hasS $ implies $ \forall \sigma : \mc \to \mc \pkt \sigma\!\left( P_1 \right) \hasS $. Note that $ \sigma\!\left( T \right) = \sigma\!\left( \Res{z}{P_1} \right) \stackrel{\text{Def. } \sigma}{=} \Res{z'}{\sigma\!\left( \subs{z}{z'}\!\left( P_1 \right) \right)} $ for all substitutions $ \sigma : \mc \to \mc $, where $ z' \notin \left( \names{P_1} \cup \names{\sigma} \right) $. Thus $ \forall \sigma : \mc \to \mc \pkt \sigma\!\left( P_1 \right) \hasS $ implies $ \forall \sigma : \mc \to \mc \pkt \sigma\!\left( T \right) \hasS $, because $ \ungSub{\sigma\!\left( \subs{z'}{z}\!\left( P_1 \right) \right)} \subseteq \ungSub{\sigma\!\left( T \right)} $.
				\item[Case $ T = \; !P_1 $:] Assume $ T \hasS $. Then $ T = \; !P_1 $ implies that $ P_1 \hasS $. By the first induction hypothesis then $ T \hasS $ implies $ \forall \sigma : \mc \to \mc \pkt \sigma\!\left( P_1 \right) \hasS $. Note that $ \sigma\!\left( T \right) = \sigma\!\left( !P_1 \right) \stackrel{\text{Def. } \sigma}{=} \; !\sigma\!\left( P_1 \right) $ for all substitutions $ \sigma : \mc \to \mc $. Thus, as in the fourth case, $ \sigma\!\left( P_1 \right) \hasS $ implies $ \sigma\!\left( T \right) \hasS $ for all substitutions $ \sigma : \mc \to \mc $.
			\end{description}
	\end{description}
\end{proof}

In both calculi substitutions may allow us to reach success by enabling a communication step. To do so it has to unify two free names that are the links of an unguarded input and an unguarded output. In the case of \piNM the enabling of such a new communication step is indeed the only possibility for a substitution to influence the reachability of success.
More precisely, if in \piNM a substitution $ \sigma $ allows to reach success, \ie if $ \sigma\!\left( T \right) \reachS $ but $ T \nReachS $, then there is a derivative of $ T $ in which $ \sigma $ unifies the free link names of an input and an output guard and thus enables a new communication step.
We first consider the case where $ \sigma\!\left( P \right) $ reaches success in a single step.

\begin{lemma}
	\label{lem:newCom}
	Let $ P, P' \in \procPiNoMatch $ and $ \sigma : \mc \to \mc $ such that $ \sigma\!\left( P \right) \step P' $ with $ P' \hasS $ and $ P \nReachS $. Then:
	\begin{align*}
		& \exists P'', P_1, P_2, P_3, P_4 \in \procPiNoMatch \pkt \exists y \in \mc \pkt \exists a, b \in \freeNames{P} \pkt\\
		& \hspace{1em} P \equiv P'' \wedgeL a, b \notin \boundNames{P''} \wedgeL \left( P_1 \mid P_2 \right) \in \ungSub{P''} \wedgeL \piInput{a}{y}{P_3} \in \ungSub{P_1}\\
		& \hspace{1em} \wedge \; \piOutput{b}{y}{P_4} \in \ungSub{P_2} \wedgeL \sigma(a) = \sigma(b) \wedgeL a \neq b
	\end{align*}
\end{lemma}

\begin{proof}
	By the reduction semantics in Figure~\ref{fig:reductionSemantics}, $ \sigma\!\left( P \right) \step P' $ either results from the Rule~\textbf{Com} or the Rule~\textbf{Tau}. In the second case, because of the Rules~\textbf{Sum}, \textbf{Par}, \textbf{Res}, and \textbf{Cong}, we have $ \tau.\sigma\!\left( Q \right) \in \ungSub{\sigma\!\left( P \right)} $ for some $ Q \in \procPiNoMatch $ and $ P' \equiv \sigma\!\left( Q \right) $. But then $ P \step Q $ and, because $ P' \hasS $ implies $ Q \hasS $ (see Lemma~\ref{prop:propequiv}), we have $ P \reachS $. This contradicts our assumptions.
	
	If the step $ \sigma\!\left( P \right) \step P' $ results from the Rule~\textbf{Com} then, because of the Rules~\textbf{Sum}, \textbf{Par}, \textbf{Res}, and \textbf{Cong}, there are $ P'', P_1, P_2, Q_1, Q_2 \in \procPiNoMatch $ and $ x, y \in \mc $ such that $ \sigma\!\left( P \right) \equiv \sigma\!\left( P'' \right) $, $ \left( \sigma\!\left( P_1 \right) \mid \sigma\!\left( P_2 \right) \right) \in \ungSub{\sigma\!\left( P'' \right)} $, $ \sigma\!\left( P_1 \right) \labeledStep{\piIn{x}{y}} Q_1 $, $ \sigma\!\left( P_2 \right) \labeledStep{\piOut{x}{y}} Q_2 $, and $ y \notin \names{\sigma} $.
	Accordingly, also $ P \equiv P'' $ and $ \left( P_1 \mid P_2 \right) \in \ungSub{P''} $.
	Without loss of generality assume $ \boundNames{P''} \cap \freeNames{P''} = \emptyset $.
	Note that the labeled steps have to result from the Rules~\textbf{Input} and~\textbf{Output}, and arbitrary many applications of \textbf{Sum-l}, \textbf{Par-l}, and \textbf{Cong-l} but no other rule. Hence there are $ R_1, R_2 \in \procPiNoMatch $ such that $ \piInput{x}{y}{\sigma\!\left( R_1 \right)} \in \ungSub{\sigma\!\left( P_1 \right)} $ and $ \piOutput{x}{y}{\sigma\!\left( R_2 \right)} \in \ungSub{\sigma\!\left( P_2 \right)} $ and we can choose $ P_3 = \sigma\!\left( R_1 \right) $ and $ P_4 = \sigma\!\left( R_2 \right) $.
	Then $ P' \equiv \sigma\!\left( R_1 \right) \mid \sigma\!\left( R_2 \right) = \sigma\!\left( R_1 \mid R_2 \right) $. Because it is not allowed to use the Rule~\textbf{Res} for the labeled steps, the $ x $ in $ \piInput{x}{y}{\sigma\!\left( R_1 \right)} $ is restricted in $ \sigma\!\left( P \right) $ iff the $ x $ in $ \piOutput{x}{y}{\sigma\!\left( R_2 \right)} $ is restricted in $ \sigma\!\left( P \right) $. If both occurrences of $ x $ are restricted in $ \sigma\!\left( P \right) $ then there are corresponding restricted names in $ P $. But then $ P \step R_1 \mid R_2 $ and, because $ P' \hasS $ implies $ \left( R_1 \mid R_2 \right) \hasS $ (see Lemma~\ref{prop:propequiv}), we have $ P \reachS $. This contradicts our assumptions.
	
	Hence these two occurrences of $ x $ are free in $ \sigma\!\left( P \right) $ and there are $ a, b \in \freeNames{P} $ such that $ \sigma\!\left( a \right) = \sigma\!\left( b \right) = x $ and $ a, b \notin \boundNames{P''} $. If $ a = b $ then again $ P \step R_1 \mid R_2 $ and, because $ P' \hasS $ implies $ \left( R_1 \mid R_2 \right) \hasS $ (see Lemma~\ref{prop:propequiv}), we have $ P \reachS $. This contradicts our assumptions.
	
	Otherwise, we have $ a \neq b $ and are done.
\end{proof}

From this, we can show the statement for arbitrarily many steps that are
necessary to reach success. Note that, because of the reduction
Rule~\textbf{Cong}, $ P \equiv\steps P' $ is short for $ P \equiv P' $ or $
P \steps P' $.

\begin{lemma}
	\label{prop:newCom}
	Let $ P \in \procPiNoMatch $ and $ \sigma : \mc \to \mc $ such that $ \sigma\!\left( P \right) \reachS $ but $ P \nReachS $.
	Then:
	\begin{align*}
		& \exists P', P_1, P_2, P_3, P_4 \in \procPiNoMatch \pkt \exists y \in \mc \pkt \exists a, b \in \freeNames{P} \pkt\\
		& \hspace{1em} P \equiv\steps P' \wedgeL a, b \notin \boundNames{P'} \wedgeL \left( P_1 \mid P_2 \right) \in \ungSub{P'} \wedgeL \piInput{a}{y}{P_3} \in \ungSub{P_1}\\
		& \hspace{1em} \wedge \; \piOutput{b}{y}{P_4} \in \ungSub{P_2} \wedgeL \sigma(a) = \sigma(b) \wedgeL a \neq b
	\end{align*}
\end{lemma}

\begin{proof}
	If $ \sigma\!\left( P \right) \hasS $ then, by Lemma~\ref{prop:propequiv}, also $ P \hasS $. Hence there are $ Q_1, Q_2 \in \procPiNoMatch $ such that $ \sigma\!\left( P \right) \steps Q_1 \step Q_2 $, $ Q_1 \nHasS $, and $ Q_2 \hasS $. By Lemma~\ref{lem:newCom}, then $ Q_1 \step Q_2 $, $ Q_1 \nHasS $, and $ Q_2 \hasS $ imply that there are $ P', P_1, P_2, P_3, P_4, a, b $ and $ y $ as required.
\end{proof}

In the following proofs we often use the term $ \match{a}{b}\success $ or a variant of this term as counterexample. To reason about the encoding of this term we analyse the context $ \contextMatch{a}{b}{N \cup \Set[]{ a, b }}{\hole} $ that is introduced according to compositionality to translate $ \match{a}{b} $. Note that this context is parameterised on $ N \cup \Set{ a, b } $, which is the set of free names of the encoded term. For example in the case of $ \match{a}{b}\success $ the set of free names contains only $ a $ and $ b $, \ie $ N = \emptyset $. Moreover $ a $, $ b $ and the continuation of the match prefix are parameters of this context. First we show that this context cannot reach success on its own, \ie without a term in its hole.

\begin{lemma}
	\label{lem:contextCannotReachSuccess}
	Let \encod be a valid encoding from \piT into \piNM.
	Let $ N \subseteq \mc $ be a finite set of names,
	$ a, b \in \mc $ be names such that $ a \neq b $,
	and $ \contextMatch{a}{b}{N \cup \Set[]{ a, b }}{\hole} $ be the context that is introduced by $ \enco{\cdot} $ to encode the match prefix $ \match{a}{b} $. Then $ \contextMatch{a}{b}{N \cup \Set[]{ a, b }}{\hole} $ cannot reach success on its own, \ie $ \contextMatch{a}{b}{N \cup \Set[]{ a, b }}{\hole} \nReachS $.
\end{lemma}

\begin{proof}
	Assume the contrary, \ie there is a context $ \context{C'}{\hole} : \procPiNoMatch \to \procPiNoMatch $ such that $ \contextMatch{a}{b}{N \cup \Set[]{ a, b }}{\hole} \steps \context{C'}{\hole} $ and $ \context{C'}{\hole} \hasS $, \ie $ \context{C'}{\hole} $ has an unguarded occurrence of success. Let $ P \in \procPiNoMatch $ such that $ \freeNames{P} = N $. Then $ \enco{\match{a}{b}P} = \contextMatch{a}{b}{N \cup \Set[]{ a, b }}{\enco{P}} \steps \context{C'}{\enco{P}} $ and, because of $ \context{C'}{\hole} \hasS $, we also have $ \context{C'}{\enco{P}} \hasS $. Hence $ \enco{\match{a}{b}P} \reachS $ but $ \left( \match{a}{b}P \right) \nReachS $, because of $ a \neq b $. This contradicts success sensitiveness.
\end{proof}

Moreover the context introduced to encode the match prefix has to ensure that its hole, \ie the respective encoding of the continuation of the match prefix, is initially guarded and cannot be unguarded by the context on its own.

\begin{lemma}
	\label{lem:contextNotUnguardContinuation}
	Let \encod be a valid encoding from \piT into \piNM.
	Let $ N \subseteq \mc $ be an arbitrary finite set of names,
	$ a, b \in \mc $ be arbitrary names such that $ a \neq b $,
	and $ \contextMatch{a}{b}{N \cup \Set[]{ a, b }}{\hole} $ be the context that is introduced by $ \enco{\cdot} $ to encode the match prefix $ \match{a}{b} $. Then $ \contextMatch{a}{b}{N \cup \Set[]{ a, b }}{\hole} $ cannot unguard its hole, \ie $ \contextMatch{a}{b}{N \cup \Set[]{ a, b }}{\hole} \steps \context{C'}{\hole} $ implies $ \hole \notin \ungSub{\context{C'}{\hole}} $ for all $ \context{C'}{\hole} : \procPiNoMatch \to \procPiNoMatch $.
\end{lemma}

\begin{proof}
	Assume the contrary, \ie assume there is $ \context{C'}{\hole} \in \procPiNoMatch \to \procPiNoMatch $ such that $ \contextMatch{a}{b}{N \cup \Set[]{ a, b }}{\hole} \steps \context{C'}{\hole} $ and $ \hole \in \ungSub{\context{C'}{\hole}} $. Then, by Lemma~\ref{prop:subsetUngSub}, $ \ungSub{\hole} \subseteq \ungSub{\context{C'}{\hole}} $. Let $ P \in \procPi $ be a source term such that $ \freeNames{P} = N $. Then $ \freeNames{P \mid \success} = N $ and $ \enco{\match*{a}{b}{P \mid \success}} = \contextMatch{a}{b}{N \cup \Set[]{ a, b }}{\enco{P \mid \success}} \steps \context{C'}{\enco{P \mid \success}} $ and $ \ungSub{\enco{P \mid \success}} \subseteq \ungSub{\context{C'}{\enco{P \mid \success}}} $. Since $ \left( P \mid \success \right) \reachS $ and by success sensitiveness, we have $ \enco{P \mid \success} \reachS $. Hence $ \enco{\match*{a}{b}{P \mid \success}} \reachS $ but $ \left( \match*{a}{b}{P \mid \success} \right) \nReachS $, because of $ a \neq b $. This contradicts success sensitiveness.
\end{proof}

But as soon as a substitution unifies the names $ a $ and $ b $ within $ \match{a}{b} $ the context has to unguard its hole.

\begin{lemma}
	\label{lem:subContextUnguardsHole}
	Let \encod be a valid encoding from \piT into \piNM.
	Let $ N \subseteq \mc $ be a finite set of names,
	$ a, b \in \mc $ be names such that $ a \neq b $,
	let $ \contextMatch{a}{b}{N \cup \Set[]{ a, b }}{\hole} $ be the context that is introduced by $ \enco{\cdot} $ to encode the match prefix $ \match{a}{b} $,
	and let $ \sigma : \mc \to \mc $ be a substitution such that $ \sigma(a) = \sigma(b) $.
	Then
	\begin{align*}
		\exists \context{C'}{\hole} \in \procPiNoMatch \to \procPiNoMatch \pkt \sigma'\!\left( \contextMatch{a}{b}{N \cup \Set[]{ a, b }}{\hole} \right) \steps \context{C'}{\hole} \wedge \hole \in \ungSub{\context{C'}{\hole}}
	\end{align*}
	where $ \sigma' $ is such that $ \vap\!\left( \sigma\!\left( n \right) \right) = \sigma'\!\left( \vap\!\left( n \right) \right) $ for every $ n \in \mc $.
\end{lemma}

\begin{proof}
	Consider the terms $ \match{a}{b}\nullTerm $ and $ \match{a}{b}\success $. By name invariance, $ \enco{\sigma\!\left( \match{a}{b}\nullTerm \right)} \asymp \sigma'\!\left( \enco{\match{a}{b}\nullTerm} \right) = \sigma'\!\left( \contextMatch{a}{b}{N \cup \Set[]{ a, b }}{\enco{\nullTerm}} \right) $ and $ \enco{\sigma\!\left( \match{a}{b}\success \right)} \asymp \sigma'\!\left( \enco{\match{a}{b}\success} \right) = \sigma'\!\left( \contextMatch{a}{b}{N \cup \Set[]{ a, b }}{\enco{\success}} \right) $. If $ \sigma'\!\left( \contextMatch{a}{b}{N \cup \Set[]{ a, b }}{\hole} \right) $ never unguards its hole, \ie if there is no $ \context{C'}{\hole} \in \procPiNoMatch \to \procPiNoMatch $ such that $ \hole \in \ungSub{\context{C'}{\hole}} $ and $ \sigma'\!\left( \contextMatch{a}{b}{N \cup \Set[]{ a, b }}{\hole} \right) \steps \context{C'}{\hole} $, then $ \sigma'\!\left( \contextMatch{a}{b}{N \cup \Set[]{ a, b }}{\enco{\nullTerm}} \right) \reachS $ holds if and only if $ \sigma'\!\left( \contextMatch{a}{b}{N \cup \Set[]{ a, b }}{\enco{\success}} \right) \reachS $. Since $ \asymp $ respects success, then also $ \enco{\sigma\!\left( \match{a}{b}\nullTerm \right)} \reachS $ iff $ \enco{\sigma\!\left( \match{a}{b}\success \right)} \reachS $. But, since $ \sigma\!\left( \match{a}{b}\nullTerm \right) \nReachS $ and $ \sigma\!\left( \match{a}{b}\success \right) \reachS $, this contradicts success sensitiveness.
\end{proof}

Next we combine our knowledge of the context $ \contextMatch{a}{b}{N \cup \Set[]{ a, b }}{\hole} $ and the relationship between substitutions and the reachability of success in \piNM as stated in Lemma~\ref{prop:newCom}. We show that a match prefix $ \match{a}{b} $ has to be translated into two communication partners, \ie an input and an output, on the translations of $ a $ and $ b $. Intuitively such a communication is the only way to simulate the test for equality of names that is performed by $ \match{a}{b} $. Moreover we derive that the respective links of the communication partner have to be free in the context $ \contextMatch{a}{b}{N \cup \Set[]{ a, b }}{\hole} $. Intuitively they have to be free, because otherwise no substitution can unify them. Consider for example the term $ \piOutput{x}{a}{\nullTerm} \mid \piInput{x}{b}{\match{a}{b}\success} $. In order to reach success the term first communicates the name $ a $ on $ x $. This communication leads to a substitution of the name $ b $ by the received value $ a $ in the continuation of the input guarded subterm. Only this substitution allows to unguard the only occurrence of success. Hence, to simulate such a behaviour of source terms, the encoding has to translate match prefixes into communication {partners}|{to} simulate the test for {equality}|{and} the links of these communication partners have to be {free}|{to} allow for substitutions induced by communication steps. Note that name invariance allows us to ignore a surrounding {communication}|{as} the step on link $ x $ in the example $ \piOutput{x}{a}{\nullTerm} \mid \piInput{x}{b}{\match{a}{b}\success} $|{and} to concentrate directly on the induced substitution.

To avoid the use of the criterion name invariance it suffices to show that the encodings of $ \piOutput{x}{a}{\nullTerm} \mid \piInput{x}{b}{\match{a}{b}\success} $ and $ \piOutput{x}{b}{\nullTerm} \mid \piInput{x}{b}{\match{a}{b}\success} $ differ only by a substitution of (parts of) the translations of $ a $ and $ b $\footnote{Unfortunately, because of the formulation of compositionality that allows for the contexts to depend on the free names of the term, this task is technically elaborate.} and that the contexts introduced to encode outputs and inputs cannot lead to success themselves, \ie that the encoding of $ \piOutput{x}{a}{\nullTerm} \mid \piInput{x}{b}{\match{a}{b}\success} $ reaches success iff the encoding of $ \match{a}{b}\success $ is unguarded. This suffices to reconstruct the substitution and the conditions on this substitution that are used in Lemma~\ref{lem:matchReqCom} and to prove Lemma~\ref{prop:newCom} and Lemma~\ref{lem:parContextNotRes} \wrt such substitutions. The remaining proofs remain the same.

\begin{lemma}
	\label{lem:matchReqCom}
	Let \encod be a valid encoding from \piT into \piNM.
	Let $ N \subseteq \mc $ be an arbitrary finite set of names,
	$ a, b \in \mc $ be arbitrary names such that $ a \neq b $,
	and $ \contextMatch{a}{b}{N \cup \Set[]{ a, b }}{\hole} $ be the context that is introduced by $ \enco{\cdot} $ to encode the match prefix $ \match{a}{b} $.
	Then:
	\begin{align*}
		& \exists i \in \Set{ 1, \ldots, \length{\vap\!\left( a \right)} } \pkt\\
		& \hspace{3em} \Big( \exists \context{C'}{\hole}, \context{C_1}{\hole}, \context{C_2}{\hole}, \context{C_3}{\hole}, \context{C_4}{\hole} \in \procPiNoMatch \to \procPiNoMatch \pkt \exists y \in \mc \pkt \exists c, d \in \freeNames{\contextMatch{a}{b}{N \cup \Set[]{ a, b }}{\hole}} \pkt\\
		& \hspace{5em} \contextMatch{a}{b}{N \cup \Set[]{ a, b }}{\hole} \equiv\steps \context{C'}{\hole}
		\wedgeL \left( \context{C_1}{\hole} \mid \context{C_2}{\hole} \right) \in \ungSub{\context{C'}{\hole}}\\
		& \hspace{5em} \wedge \; \piInput{c}{y}{\context{C_3}{\hole}} \in \ungSub{\context{C_1}{\hole}}
		\wedgeL \piOutput{d}{y}{\context{C_4}{\hole}} \in \ungSub{\context{C_2}{\hole}}\\
		& \hspace{5em} \wedge \; \Set{ c, d } = \Set{ \proj{\vap\!\left( a \right)}{i}, \proj{\vap\!\left( b \right)}{i} } \Big)\\
		& \hspace{3em} \wedge \; \Big( \forall \sigma' : \mc \to \mc \pkt \forall \context{C'}{\hole} \in \procPiNoMatch \to \procPiNoMatch \pkt\\
		& \hspace{5em} \left( \sigma'\left( \contextMatch{a}{b}{N \cup \Set[]{ a, b }}{\hole} \right) \steps \context{C'}{\hole} \wedgeL \hole \in \ungSub{\context{C'}{\hole}} \right)\\
		& \hspace{6em} \impl \sigma'\left( \proj{\vap\!\left( a \right)}{i} \right) = \sigma'\left( \proj{\vap\!\left( b \right)}{i} \right)\\
		& \hspace{8em} \wedge \; \text{in } \sigma'\!\left( \contextMatch{a}{b}{N \cup \Set[]{ a, b }}{\hole} \right) \steps \context{C'}{\hole} \text{ there is a step on } \sigma'\left( \proj{\vap\!\left( b \right)}{i} \right) \Big)
	\end{align*}
\end{lemma}

\begin{proof}
	Consider a source term $ \match{a}{b}\left( \success \mid P \right) $, where $ P \in \procPi $ is such that $ \freeNames{P} = N $, and the substitution $ \sigma : \mc \to \mc $ with $ \sigma = \subs{b}{a} $. Then $ \left( \match{a}{b}\left( \success \mid P \right) \right) \nReachS $ and $ \sigma\!\left( \match{a}{b}\left( \success \mid P \right) \right) \reachS $. By success sensitiveness, then $ \enco{\match{a}{b}\left( \success \mid P \right)} \nReachS $, \ie $ \contextMatch{a}{b}{N \cup \Set[]{ a, b }}{\enco{\success \mid P}} \nReachS $, and $ \enco{\sigma\!\left( \match{a}{b}\left( \success \mid P \right) \right)} \reachS $.
	By name invariance, then $ \enco{\sigma\!\left( \match{a}{b}\left( \success \mid P \right) \right)} \asymp \sigma'\!\left( \enco{\match{a}{b}\left( \success \mid P \right)} \right) = \sigma'\!\left( \contextMatch{a}{b}{N \cup \Set[]{ a, b }}{\enco{\success \mid P}} \right) $ and thus, because $ \asymp $ respects success, $ \sigma'\!\left( \contextMatch{a}{b}{N \cup \Set[]{ a, b }}{\enco{\success \mid P}} \right) \reachS $, where $ \sigma' $ is such that $ \vap\!\left( \sigma\!\left( n \right) \right) = \sigma'\!\left( \vap\!\left( n \right) \right) $ for every $ n \in \mc $.
	By Lemma~\ref{prop:newCom}, $ \sigma'\!\left( \contextMatch{a}{b}{N \cup \Set[]{ a, b }}{\enco{\success \mid P}} \right) \reachS $ and $ \contextMatch{a}{b}{N \cup \Set[]{ a, b }}{\enco{\success \mid P}} \nReachS $ imply that there are $ T', T_1, T_2, T_3, T_4 \in \procPiNoMatch $, $ y \in \mc $, and $ c, d \in \freeNames{\contextMatch{a}{b}{N \cup \Set[]{ a, b }}{\enco{\success \mid P}}} $ such that $ \contextMatch{a}{b}{N \cup \Set[]{ a, b }}{\enco{\success \mid P}} \equiv\steps T' $, $ \left( T_1 \mid T_2 \right) \in \ungSub{T'} $, $ \piInput{c}{y}{T_3} \in \ungSub{T_1} $, $ \piOutput{d}{y}{T_4} \in \ungSub{T_2} $, $ \sigma'\!\left( c \right) = \sigma'\!\left( d \right) $, and $ c \neq d $.
	
	By Lemma~\ref{lem:contextNotUnguardContinuation}, the context $ \contextMatch{a}{b}{N \cup \Set[]{ a, b }}{\hole} $ cannot unguard its hole. 
	We conclude that within $ T' $ the encoded continuation $ \enco{T} $ is still guarded and can only be unguarded in $ \sigma'\!\left( \enco{\match{a}{b}\left( \success \mid P \right)} \right) $ if $ \sigma'(c) = \sigma'(d) $ and to unguard the continuation there is a step on channel $ \sigma'(c) $.
	Thus $ c, d \in \freeNames{\contextMatch{a}{b}{N \cup \Set[]{ a, b }}{\hole}} $ and there are $ \context{C'}{\hole}, \context{C_1}{\hole}, \context{C_2}{\hole}, \context{C_3}{\hole}, \context{C_4}{\hole} \in \procPiNoMatch \to \procPiNoMatch $ such that
	\begin{itemize}
		\item $ \contextMatch{a}{b}{N \cup \Set[]{ a, b }}{\hole} \equiv\steps \context{C'}{\hole} $,
		\item $ \left( \context{C_1}{\hole} \mid \context{C_2}{\hole} \right) \in \ungSub{\context{C'}{\hole}} $,
		\item $ \piInput{c}{y}{\context{C_3}{\hole}} \in \ungSub{\context{C_1}{\hole}} $, and
		\item $ \piOutput{d}{y}{\context{C_4}{\hole}} \in \ungSub{\context{C_2}{\hole}} $.
	\end{itemize}
	
	It remains to prove $ \Set{ c, d } = \Set{ \proj{\vap\!\left( a \right)}{i}, \proj{\vap\!\left( b \right)}{i} } $. Since $ \sigma'\!\left( c \right) = \sigma'\!\left( d \right) $ but $ c \neq d $, at least one of the names $ c, d $ has to be in the domain of $ \sigma' $. Assume only one of the names|say $ c $|is in the domain of $ \sigma' $, \ie $ \sigma'\!\left( d \right) = d $. By name invariance, all names that are in the domain of $ \sigma' $ are (parts of) the translation of source term names, \ie there is some source name $ x \in \mc $ such that $ c \in \vap\!\left( x \right) $.
	Because $ \sigma' $ is such that $ \vap\!\left( \sigma\!\left( n \right) \right) = \sigma'\!\left( \vap\!\left( n \right) \right) $ for every $ n \in \mc $, then there is some $ i \in \Set{ 1, \ldots, \length{\vap\!\left( x \right)} } $ such that $ \proj{\vap\!\left( \sigma\!\left( x \right) \right)}{i} = \sigma'\!\left( \proj{\vap\!\left( x \right)}{i} \right) = \sigma'\!\left( c \right) = \sigma'\!\left( d \right) = d $. Thus also $ d $ is equal to (a part of) the translation of a source term name. The same holds if we change the roles of $ c $ and $ d $.
	Hence|regardless of whether only one or both names are in the domain of $ \sigma' $|there are two source names $ x, y \in \mc $ and $ i, j \in \Set{ 1, \ldots, n } $ with $ \length{\vap\!\left( x \right)} = n = \length{\vap\!\left( y \right)} $ such that $ c = \proj{\vap\!\left( x \right)}{i} $ and $ d = \proj{\vap\!\left( y \right)}{j} $.
	
	Because $ \sigma' $ is such that $ \vap\!\left( \sigma\!\left( n \right) \right) = \sigma'\!\left( \vap\!\left( n \right) \right) $ for every $ n \in \mc $, $ \sigma'\!\left( c \right) = \sigma'\!\left( d \right) $ implies
	\vspace{-1em}
	\begin{align*}
		\proj{\vap\!\left( \sigma\!\left( x \right) \right)}{i} = \sigma'\!\left( \proj{\vap\!\left( x \right)}{i} \right) = \sigma'\!\left( c \right) = \sigma'\!\left( d \right) = \sigma'\!\left( \proj{\vap\!\left( y \right)}{j} \right) = \proj{\vap\!\left( \sigma\!\left( y \right) \right)}{j}.
	\end{align*}
	
	\vspace{-1em}
	\noindent
	Because the renaming policy $ \vap $ is such that $ \vap\!\left( u \right) \cap \vap\!\left( v \right) = \emptyset $ whenever $ u \neq v $, the above equation, \ie $ \proj{\vap\!\left( \sigma\!\left( x \right) \right)}{i} = \proj{\vap\!\left( \sigma\!\left( y \right) \right)}{j} $, implies $ \sigma\!\left( x \right) = \sigma\!\left( y \right) $ and, because of Observation~\ref{obs:renamingPolicyVector}, $ i = j $.
	Then $ c = \proj{\vap\!\left( x \right)}{i} $, $ d = \proj{\vap\!\left( y \right)}{i} $, and $ c \neq d $ imply $ x \neq y $.
	Because $ \sigma = \subs{b}{a} $, \ie $ \sigma(a) = \sigma(b) \neq \sigma(n) $ and $ \sigma\!\left( n \right) = n $ for all $ n \in \left( \mc \setminus \Set{ a, b } \right) $, $ \sigma\!\left( x \right) = \sigma\!\left( y \right) $ and $ x \neq y $ imply $ \Set{ x, y } = \Set{ a, b } $.
	From $ \Set{ x, y } = \Set{ a, b } $, $ c = \proj{\vap\!\left( x \right)}{i} $, and $ d = \proj{\vap\!\left( y \right)}{i} $ we conclude $ \Set{ c, d } = \Set{ \proj{\vap\!\left( a \right)}{i}, \proj{\vap\!\left( b \right)}{i} } $ as required.
\end{proof}

\noindent
A very important consequence of the lemma above is the existence of the index $ i $ for all contexts $ \contextMatch{a}{b}{N \cup \Set[]{ a, b }}{\hole} $ regardless of $ N $ and of the terms that may be inserted in the hole. Note that the ``there is some'' does not necessarily imply that there is just one such $ i $. If the renaming policy splits up a source term name into several target term names then different parts of this vector can be used to simulate the test for equality. However, the above lemma states that there is at least one such $ i $, \ie at least one part of the translation of source term names is used to implement the required communication partners.

To derive the separation result we need a counterexample that combines two match prefixes in parallel.
Therefore we need some information on the context $ \contextPar{N}{\hole_1}{\hole_2} $ that is introduced by $ \enco{\cdot} $ according to compositionality to translate the parallel operator. Note that this context is parameterised on the set $ N $ that consists of the free names of the two parallel components that should be encoded. Moreover the two holes serve as placeholders for the encoding of the left and the right hand side of the source term.
Similar to the context introduced to encode the match prefix, the context that is introduced to encode the parallel operator cannot reach success on its own, \ie $ \contextPar{N}{\hole_1}{\hole_2} \nReachS $.

\begin{lemma}
	\label{lem:parContextCannotReachSuccess}
	Let \encod be a valid encoding from \piT into \piNM.
	Let $ N \subseteq \mc $ be a finite set of names
	and $ \contextPar{N}{\hole_1}{\hole_2} $ be the context that is introduced by $ \enco{\cdot} $ to encode the parallel operator.
	Then $ \contextPar{N}{\hole_1}{\hole_2} $ cannot reach success on its own, \ie:
	\begin{align*}
		\contextPar{N}{\hole_1}{\hole_2} \nReachS
	\end{align*}
\end{lemma}

\begin{proof}
	Assume the contrary, \ie there is a context $ \context{C'}{\hole_1; \hole_2} : \procPiNoMatch^2 \to \procPiNoMatch $ such that $ \contextPar{N}{\hole_1}{\hole_2} \steps \context{C'}{\hole_1; \hole_2} $ and $ \context{C'}{\hole_1; \hole_2} \hasS $. Let $ P \in \procPiNoMatch $ such that $ \freeNames{P} = N $ and $ P \nReachS $. Then $ \enco{P \mid \nullTerm} = \contextPar{N}{\enco{P}}{\enco{\nullTerm}} $, $ \contextPar{N}{\enco{P}}{\enco{\nullTerm}} \steps \context{C'}{\enco{P}; \enco{\nullTerm}} $ and, because of $ \context{C'}{\hole_1; \hole_2} \hasS $, we also have $ \context{C'}{\enco{P}; \enco{\nullTerm}} \hasS $. Hence $ \enco{P \mid \nullTerm} \reachS $ but $ \left( P \mid \nullTerm \right) \nReachS $, because $ P \nReachS $. This contradicts success sensitiveness.
\end{proof}

But in contrast to the context introduced in order to encode the match prefix the context $ \contextPar{N}{\hole_1}{\hole_2} $ has always, \ie regardless of a substitution, to be able to unguard its holes on its own.

\begin{lemma}
	\label{lem:parContextUnguard}
	Let \encod be a valid encoding from \piT into \piNM.
	Let $ N \subseteq \mc $ be a finite set of names
	and $ \contextPar{N}{\hole_1}{\hole_2} $ be the context that is introduced by $ \enco{\cdot} $ to encode the parallel operator.
	Then:
	\begin{align*}
		\exists T \in \procPiNoMatch \pkt \contextPar{N}{\hole_1}{\hole_2} \steps T \wedgeL \hole_1, \hole_2 \in \ungSub{T}
	\end{align*}
\end{lemma}

\begin{proof}
	Consider a term $ P \in \procPi $ such that $ \freeNames{P} = N $.
	Note that, by success sensitiveness, $ \nullTerm \nReachS $ implies $ \enco{\nullTerm} \nReachS $.
	Then $ \enco{\left( P \mid \success \right) \mid \nullTerm} = \contextPar{N}{\enco{P \mid \success}}{\enco{\nullTerm}} $.
	By success sensitiveness, $ \contextPar{N}{\enco{P \mid \success}}{\enco{\nullTerm}} \reachS $.
	But if $ \contextPar{N}{\hole_1}{\hole_2} $ never unguards its first hole then $ \contextPar{N}{\enco{X}}{\enco{\nullTerm}} \nReachS $ for all $ X \in \procPi $, because by Lemma~\ref{lem:parContextCannotReachSuccess} the context cannot reach success on its own and also $ \enco{\nullTerm} \nReachS $. The argumentation for the second hole is similar with $ \nullTerm \mid \left( P \mid \success \right) $.
\end{proof}

Then we need to show that the context $ \contextPar{N}{\hole_1}{\hole_2} $ cannot bind the names that are used by the context $ \contextMatch{a}{b}{N \cup \Set[]{ a, b }}{\hole} $ to simulate the test for equality. As in the above example $ \piOutput{x}{a}{\nullTerm} \mid \piInput{x}{b}{\match{a}{b}\success} $, a communication step can unify at runtime the variables of a match prefix. Such a communication step naturally transmits the value for the match variables over a parallel operator, because communication is always between two communication partners that are composed in parallel. If this value is restricted on either side of the parallel operator the communication could not lead to the required unification. The match variables would still be considered as different and the match prefix as not satisfied. Thus for example neither $ \Res*{a}{\piOutput{x}{a}{\nullTerm}} \mid \piInput{x}{b}{\match{a}{b}\success} $ nor $ \piOutput{x}{a}{\nullTerm} \mid \Res*{a}{\piInput{x}{b}{\match{a}{b}\success}} $ reach success although in both cases the communication on $ x $ is still possible. Of course the term $ \Res*{a}{\piOutput{x}{a}{\nullTerm} \mid \piInput{x}{b}{\match{a}{b}\success}} $ reaches success. But for cases like this we can construct larger counterexamples as $ \Res*{a}{\piOutput{x}{a}{\nullTerm} \mid \nullTerm} \mid \piInput{x}{b}{\match{a}{b}\success} $ and to analyse the source term, in order to examine the places at which such a restriction would be allowed, violates the idea of a compositional encoding. Again name invariance allows us to ignore the communication on $ x $ and to directly concentrate on the induced substitution.

\begin{lemma}
	\label{lem:parContextNotRes}
	Let \encod be a valid encoding from \piT into \piNM.
	Let $ N \subseteq \mc $ be a finite set of names
	and $ \contextPar{N}{\hole_1}{\hole_2} $ be the context that is introduced by $ \enco{\cdot} $ to encode the parallel operator.
	Then:
	\begin{align*}
		\forall P, Q \in \procPi \pkt \forall a \in \freeNames{P \mid Q} \pkt \proj{\vap\!\left( a \right)}{i} \in \freeNames{\contextPar{N}{\enco{P}}{\enco{Q}}}
	\end{align*}
	where $ i \in \Set{ 1, \ldots, \length{\vap\!\left( a \right)} } $ is the index that exists according to Lemma~\ref{lem:matchReqCom}.
\end{lemma}

\begin{proof}
	Consider the term $ S = \match{a}{b}\left( \success \mid P \right) \mid \nullTerm $, where $ P \in \procPi $ is such that $ \freeNames{P} = N $, and a substitution $ \sigma : \mc \to \mc $ such that $ \sigma\!\left( a \right) = \sigma\!\left( b \right) $.
	Then $ \sigma\!\left( S \right) \reachS $. Thus, by success sensitiveness, $ \enco{ \sigma\!\left( S \right) } \reachS $.
	By name invariance, $ \enco{ \sigma\!\left( S \right) } \asymp \sigma'\left( \enco{S} \right) = \sigma'\left( \contextPar{N}{\enco{\match{a}{b}\left( \success \mid P \right)}}{\enco{\nullTerm}} \right) $, where $ \vap\!\left( \sigma\!\left( n \right) \right) = \sigma'\!\left( \vap\!\left( n \right) \right) $ for every $ n \in \mc $.
	Then, because $ \asymp $ respects success, $ \sigma'\left( \contextPar{N}{\enco{\match{a}{b}\left( \success \mid P \right)}}{\enco{\nullTerm}} \right) \reachS $.
	
	By Lemma~\ref{lem:parContextCannotReachSuccess} and Observation~\ref{obs:subsSteps}, $ \sigma'\!\left( \contextPar{N}{\hole_1}{\hole_2} \right) $ cannot reach success. But, by Lemma~\ref{lem:parContextUnguard} and Observation~\ref{obs:subsSteps}, it holes can become unguarded. Also $ \enco{\nullTerm} \nReachS $, because $ \nullTerm \nReachS $ and because of success sensitiveness.
	Thus $ \enco{\match{a}{b}\left( \success \mid P \right)} = \contextMatch{a}{b}{N \cup \Set[]{ a, b }}{\enco{\success \mid P}} $ has to reach success within the term $ \sigma'\left( \contextPar{N}{\enco{\match{a}{b}\left( \success \mid P \right)}}{\enco{\nullTerm}} \right) $.
	By Lemma~\ref{lem:contextCannotReachSuccess}, also $ \contextMatch{a}{b}{N \cup \Set[]{ a, b }}{\hole} $ cannot reach success and, by Lemma~\ref{lem:matchReqCom}, the context $ \contextMatch{a}{b}{N \cup \Set[]{ a, b }}{\hole} $ can only unguard its hole if $ \sigma'\left( \proj{\vap\!\left( a \right)}{i} \right) = \sigma'\left( \proj{\vap\!\left( b \right)}{i} \right) $. Note that we choose an arbitrary $ a $ and $ b $ here. Thus, since substitutions cannot rename restricted names and because $ \sigma'\left( \contextPar{N}{\enco{\match{a}{b}\left( \success \mid P \right)}}{\enco{\nullTerm}} \right) \reachS $, the context $ \contextPar{N}{\hole_1}{\hole_2} $ cannot restrict the $ i $'th part of the translation of a name.
\end{proof}

Finally we show that there is no valid encoding from \piT into \piNM, by assuming the contrary and deriving a contradiction. As already mentioned, we use a counterexample that consists of two parallel composed match prefixes. More precisely we use $ \match{a}{b}\success \mid \match{b}{a}\success $, \ie swap the match variables on the right side. Intuitively the contradiction is derived as follows: Since the context $ \contextMatch{a}{b}{N \cup \Set[]{ a, b }}{\hole} $ translates the match variables into free links of unguarded communication partners and because of the swapping of the matching variables on the right side, the parallel composition of the two variants of the context $ \contextMatch{a}{b}{N \cup \Set[]{ a, b }}{\hole} $|that are necessary to encode the counterexample|enable wrong communication steps between a communication partner from the left $ \contextMatch{a}{b}{N \cup \Set[]{ a, b }}{\hole} $ and a communication partner from the right $ \contextMatch{b}{a}{N \cup \Set[]{ a, b }}{\hole} $. We denote such a communication step as wrong, because in this case the communication cannot lead to the unguarding of the encoded continuation $ \enco{\success} $ without violating success sensitiveness. In order to reach success, the source term needs a substitution $ \sigma $ that unifies the match variables. Unfortunately, the same wrong communication can consume one of the communication partners in $ \enco{\sigma\left( \match{a}{b}\success \mid \match{b}{a}\success \right)} $ that is necessary to unguard the encoded continuation. A restoration of this communication partner leads by symmetry to divergence which violates the divergence reflection criterion. But without the possibility of a restoration, the wrong communication leads to an unsuccessful execution of $ \enco{\sigma\left( \match{a}{b}\success \mid \match{b}{a}\success \right)} $. This execution violates the combination of success sensitiveness and operational soundness.

\begin{theorem}
	\label{thm:noEnc}
	There is no valid encoding from \piT into \piNM.
\end{theorem}

\begin{proof}
	Assume the contrary, \ie there is a valid encoding \encod from \piT into \piNM.
	Consider the term $ S = \match{a}{b}\success \mid \match{b}{a}\success $ and a substitution $ \sigma : \mc \to \mc $ such that $ \sigma(a) = \sigma(b) $.
	By success sensitiveness and because $ \enco{S} = \contextPar{\Set[]{ a, b }}{\contextMatch{a}{b}{\Set[]{ a, b }}{\enco{\success}}}{\contextMatch{b}{a}{\Set[]{ a, b }}{\enco{\success}}} $, $ S \nReachS $ implies that $ \contextPar{\Set[]{ a, b }}{\contextMatch{a}{b}{\Set[]{ a, b }}{\enco{\success}}}{\contextMatch{b}{a}{\Set[]{ a, b }}{\enco{\success}}} \nReachS $.
	By Lemma~\ref{lem:parContextUnguard}, there is some $ T $ such that $ \enco{S} \steps T $ and $ \contextMatch{a}{b}{\Set[]{ a, b }}{\enco{\success}}, \contextMatch{b}{a}{\Set[]{ a, b }}{\enco{\success}} \in \ungSub{T} $.
	Because $ \freeNames{\match{a}{b}P} = \freeNames{\match{b}{a}P} $ for all $ P $, $ \contextMatch{a}{b}{\Set[]{ a, b }}{\hole} = \Set{ \subst{\vap(b)}{\vap(a)}, \subst{\vap(a)}{\vap(b)} }\!\left( \contextMatch{b}{a}{\Set[]{ a, b }}{\hole} \right) $.
	By Lemma~\ref{lem:matchReqCom}, then there are $ i \in \Set{ 1, \ldots, \length{\vap(a)}} $, $ T' \in \procPiNoMatch $, $ \context{C_1}{\hole}, \context{C_2}{\hole}, \context{C_3}{\hole}, \context{C_4}{\hole} \in \procPiNoMatch \to \procPiNoMatch $, and $ y \in \mc $ such that $ \enco{S} \equiv\steps T' $ and the terms $ \piInput{\proj{\vap\!\left( a \right)}{i}}{y}{\context{C_1}{\enco{\success}}} $, $ \piOutput{\proj{\vap\!\left( a \right)}{i}}{y}{\context{C_2}{\enco{\success}}} $, $ \piInput{\proj{\vap\!\left( b \right)}{i}}{y}{\context{C_3}{\enco{\success}}} $, and $ \piOutput{\proj{\vap\!\left( b \right)}{i}}{y}{\context{C_4}{\enco{\success}}} $ are unguarded subterms of $ T' $.
	
	By name invariance, $ \enco{\sigma\!\left( S \right)} \asymp \sigma'\left( \enco{S} \right) = \sigma'\left( \contextPar{\Set[]{ a, b }}{\contextMatch{a}{b}{\Set[]{ a, b }}{\enco{\success}}}{\contextMatch{b}{a}{\Set[]{ a, b }}{\enco{\success}}} \right) $, where $ \vap\!\left( \sigma\!\left( n \right) \right) = \sigma'\!\left( \vap\!\left( n \right) \right) $ for every $ n \in \mc $.
	By Observation~\ref{obs:subsSteps}, $ \sigma'\left( \enco{S} \right) \equiv\steps \sigma'\!\left( T' \right) $, \ie again the inputs on the channels $ \proj{\vap\!\left( a \right)}{i} $ and $ \proj{\vap\!\left( b \right)}{i} $ can communicate between the two instances of the context $ \contextMatch{\cdot}{\cdot}{\Set[]{ a, b }}{\cdot} $ in $ \sigma'\!\left( T' \right) $.
	By the argumentation above these communications, \ie there an input from the left $ \contextMatch{a}{b}{\Set[]{ a, b }}{\enco{\success}} $ interacts with an output from the right $ \contextMatch{b}{a}{\Set[]{ a, b }}{\enco{\success}} $ or vice versa, cannot lead to the unguarding of $ \enco{\success} $.
	Note that if either the left $ \contextMatch{a}{b}{\Set[]{ a, b }}{\enco{\success}} $ or the right $ \contextMatch{b}{a}{\Set[]{ a, b }}{\enco{\success}} $ restores a wrongly consumed input term or output term on $ \proj{\vap\!\left( a \right)}{i} $ or $ \proj{\vap\!\left( b \right)}{i} $ then, because $ \contextMatch{a}{b}{\Set[]{ a, b }}{\hole} = \Set{ \subst{\vap(b)}{\vap(a)}, \subst{\vap(a)}{\vap(b)} }\!\left( \contextMatch{b}{a}{\Set[]{ a, b }}{\hole} \right) $, there is an execution there the other side also restores the corresponding counterpart. This leads back to the state before the respective communication step between the left $ \contextMatch{a}{b}{\Set[]{ a, b }}{\enco{\success}} $ and the right $ \contextMatch{b}{a}{\Set[]{ a, b }}{\enco{\success}} $ and thus to a divergent execution.
	The same holds if the context $ \contextPar{\Set[]{ a, b }}{\hole_1}{\hole_2} $ restores such an input term or output term.
	But since $ \sigma\!\left( S \right) $ has no divergent execution and because of divergence reflection, a divergent execution of $ \enco{\sigma\!\left( S \right)} $ violates our assumption that \encod is a valid encoding.
	Thus $ \sigma'\!\left( \enco{S} \right) $ cannot restore a wrongly consumed input term or output term on $ \proj{\vap\!\left( a \right)}{i} $ or $ \proj{\vap\!\left( b \right)}{i} $.
	
	By Lemma~\ref{lem:matchReqCom}, only a communication between the terms on channels $ \proj{\vap\!\left( a \right)}{i} $ and $ \proj{\vap\!\left( b \right)}{i} $ can unguard the continuation $ \enco{\success} $. Hence there is a finite maximal execution of $ \enco{\sigma\!\left( S \right)} $ in which the continuation $ \enco{\success} $ is never unguarded.
	Thus, by Lemma~\ref{lem:contextCannotReachSuccess} and Lemma~\ref{lem:parContextCannotReachSuccess}, no success is reached in this execution, \ie $ \sigma'\left( \enco{S} \right) \not\mustReachSuccessFinite $.
	By Lemma~\ref{lem:mustSuccessRespecting}, then $ \enco{\sigma\!\left( S \right)} \asymp \sigma'\left( \enco{S} \right) $ implies $ \enco{\sigma\!\left( S \right)} \not\mustReachSuccessFinite $.
	But, by Lemma~\ref{lem:mustSuccessSensitiveness}, $ \sigma\!\left( S \right) \mustReachSuccessFinite $ implies $ \enco{\sigma\!\left( S \right)} \mustReachSuccessFinite $.
	This is a contradiction.
\end{proof}


\section{Discussion}
\label{sec:discussion}

As mentioned above, also Carbone and Maffeis show in \cite{carbone} that the match prefix cannot be encoded within the \piCal. Moreover there are different encodings of the match prefix in modified variants and extensions of the \piCal. In this section we discuss the relation between these results and our separation result.

\subsection{The Match Prefix is a Native Operator of the Pi-Calculus}
\label{sec:comparison}

If we compare the approach in \cite{carbone} with ours, we observe that the considered variants of the \piCal are different.
We consider the full \piCal and its variant without the match prefix as source and target language.
In the literature there are different variants called ``full'' \piCal. We decide on the most general of these variants. In particular we consider a variant of the \piCal with free choice whereas \cite{carbone} allow only guarded choice in their target language.
Note that the source language considered in \cite{carbone} is an asynchronous variant of the \piCal, \ie is less expressive than the source language considered here \cite{palamidessi03,petersNestmann14,petersNestmannGoltz13}.
However, the only (counter)examples we use here are of the form $ \match{a}{b}X $, or $\match{a}{b}X \mid \match{b}{a}X$ where $ X $ is a combination of $ \success $, $ \nullTerm $, $ P $, and parallel composition for an arbitrary $ P $ with a fixed set of free names. Thus, our separation result remains valid if we change the source language to the asynchronous variant of the \piCal without choice that is used in \cite{carbone}.
Our target language is also more expressive, because we do not restrict it to guarded choice. More precisely, in \cite{carbone} the \piCal with guarded mixed choice is used. Accordingly, the current result can be considered stronger. However, concentrating only on guarded choice is commonly accepted and, more importantly, it might be easy to adapt the proof in \cite{carbone} to the more expressive target language.

\paragraph{Contribution 1.}
The main difference between the two approaches are the quality criteria, \ie in the conditions that are assumed to hold for all valid encodings.
Similar to \cite{palamidessi03}, Carbone and Maffeis require that an encoding must be uniform and reasonable.
By \cite{carbone} an encoding $\enco{\cdot}$ is \emph{uniform} if it translates the parallel operator homomorphically, \ie $ \enco{P \mid Q} = \enco{P}\mid \enco{Q} $, and if it respects permutations on free names, \ie for all $ \sigma $ there is some $ \theta $ such that $ \enco{\sigma(P)} = \theta(\enco{P}) $.
A \emph{reasonable} semantics, by \cite{carbone}, is one which distinguishes two processes $ P $ and $ Q $ whenever there exists a maximal execution of $ Q $ in which the observables are different from the observables in any maximal execution of $ P $.
Furthermore they require that an encoding should be able to distinguish deadlocks from livelocks, which is comparable to divergence reflection.

In contrast to uniformity, name invariance relates the substitution on the source term names with its translation on target term names. Already \cite{gorla} points out that name invariance is a more complex requirement than the above condition; but \cite{gorla} also argues that it is rather more detailed than more demanding. Moreover we claim that name invariance is not crucial for the above separation result.
The first condition of uniformity is a strictly stronger requirement than compositionality for the parallel operator as it is discussed for instance in \cite{petphd}. However the proof in \cite{carbone} does not use the homomorphic translation of the parallel operator.

The criterion on the reasonable semantics used in \cite{carbone} is even more demanding than the first part of uniformity. It states that a source term and its encoding reach exactly the same observables.
It completely ignores the possibility to translate a source term name into a sequence of names or to simulate a source term observable by a set of target term observables even if there is a bijective mapping between an observable and its translation. The proof in \cite{carbone} makes strongly use of this criterion; exploiting the fact that the match variables are free in the match prefix.
Gorla suggests success sensitiveness and operational correspondence instead. Note that we use operational correspondence|or more precisely soundness|only in the last step of the proof to argument that if a source term reaches success in all finite maximal executions its encoding does alike. Hence, for the presented case, the combination of operational soundness and success sensitiveness is a considerably weaker requirement than the variant of reasonableness.

Overall we conclude that, because of the large difference between success sensitiveness and the variant of reasonableness considered in \cite{carbone}, our set of criteria is considerably weaker and thus the presented result is strictly stronger.

\paragraph{Contribution 2.}
The proof in \cite{carbone} is, due to the stricter criteria, shorter and easier to follow than ours. But it also reveals less information on the reason for the separation result. In contrast, the presented approach reflects the intuition that communication is close to the behaviour of the match prefix. We show that among the native operators of the \piCal input and output are the only operators close enough to possibly encode the match operator, where the link names result from the translation of the match variables. But it also reveals the reason why communication is not strong enough.
Translated match variables have to be free in the encoding of the match {prefix}|{to} allow for a guarding input to receive a value for a match {variable}|{but} they also have to be {bound}|{to} avoid unintended interactions between the translated match variables of parallel match encodings.
The other \piCal operators cannot simulate this kind of binding.

\subsection{Encodings of the Match Prefix in Pi-Like Calculi}
\label{sec:encodeMatchInOtherCalculi}

As mentioned in the introduction there are some modifications and extensions of the \piCal that allow for the encoding of the match prefix. We briefly discuss four different approaches and their relation to our separation result.

In \cite{bodeiDeganoPriami05} the input prefix $ \piIn{x}{z} $ of the \piCal is replaced by a selective input $ \piIn{x}{z \in V} $. A term guarded by $ \piIn{x}{z \in V} $ and a term guarded by a matching output prefix $ \piOut{x}{y} $ can communicate (if they are composed in parallel and) only if the transmitted value $ y $ is contained in  the set $ V $ of names specified in the selective input prefix. Accordingly selective input can be used as a conditional guard. As pointed out in \cite{bodeiDeganoPriami05}, selective input allows to encode a match prefix $ \match{a}{b} P $ simply by $ \Res{x}{\left( \piOut{x}{a} \mid \piInput{x}{y \in \Set[]{ b }}{\left\lbr P \right\rbr} \right)} $, where $ \left\lbr P \right\rbr $ is the encoding of $ P $. Here the test for equality $ a = b $ is transferred into the test $ a \in \Set[]{ b } $. Thus it is not necessary to translate the match variables into communication channels, which allows for this simple encoding.

Mobile ambients \cite{cardelliGordon00} extend the asynchronous \piCal with ambients $ n\!\left[ \; \right] $, \ie sides or locations, that
\begin{inparaenum}[(a)]
	\item can contain processes and other ambients,
	\item can be composed in parallel to other ambients and processes, and
	\item whose name can be restricted to forbid interaction with its environment.
\end{inparaenum}
Moreover there are three additional actions prefixes:
\begin{inparaenum}[(1)]
	\item $ \textsf{in} \, n $ allows an ambient to enter another ambient named $ n $ by the rule $ m\!\left[ \textsf{in} \, n.P \mid Q \right] \mid n\!\left[ R \right] \step n\!\left[ m\!\left[ P \mid Q \right] \mid R \right] $,
	\item $ \textsf{out} \, n $ allows an ambient to exit its own parent named $ n $ by the rule $ n\!\left[ m\!\left[ \textsf{out} \, n.P \mid Q \right] \mid R \right] \step m\!\left[ P \mid Q \right] \mid n\!\left[ R \right] $, and
	\item $ \textsf{open} \, n $ dissolves an ambient with name $ n $ by the rule $ \textsf{open} \, n.P \mid n\!\left[ Q \right] \step P \mid Q $.
\end{inparaenum}
As a consequence, communication steps become locale, \ie can occur only if both communication partners are located in parallel within the same ambient. Hence channel names become superfluous, since communications on different channels can be simulated by communications within different ambients. So the $ \pi $-input $ \piInput{x}{z}{P} $ is replaced by $ \left( z \right)\!.P $ and the asynchronous output $ \piOut{x}{y} $ is replaced by $ \left\langle y \right\rangle $.
As pointed out in \cite{vig}, mobile ambients can encode the match prefix. They suggest to encode a match prefix $ \match{a}{b}P $ by the term $ M = \Res*{xy}{x\!\left[ \textsf{open} \, a.y\!\left[ \textsf{out} \, x \right] \mid b\!\left[ \; \right] \right] \mid \textsf{open} \, y.\textsf{open} \, x.\left\lbr P \right\rbr} $, where $ \left\lbr P \right\rbr $ is the encoding of $ P $.
Since there are no channel names, the match variables are translated into the new capabilities of mobile ambients, namely into $ \textsf{open} \, a $ and an ambient with name $ b $. $ \textsf{open} \, a $ can only be reduced if $ a = b $, \ie if either $ a = b $ holds from the beginning or if $ a $ and $ b $ are unified by a substitution induced by a surrounding input, as \eg in $ \left( b \right)\!.M \mid \left\langle a \right\rangle $. Note that, to enable this substitution, the match variables $ a $ and $ b $ {have}|{as} shown in our proof {above}|{to} be translated into free names. Here the ambient $ x $ and its restriction ensure that there are no unintended interactions between the translated match variables of parallel match encodings. More precisely the ambient $ x $ encapsulates the translation of the test for equality $ a = b $ and the restriction $ \Res{x}{} $ ensures that the environment cannot interfere, \ie no other action on the names $ a $ or $ b $ can reduce the $ \textsf{open} \, a $ or can target the ambient $ b $ inside of $ x $, because the restriction forbids other processes to enter $ x $. So in mobile ambients it is not necessary to translate the match variables into bound names, which allows for the encoding.

\cite{vivas} extend the pi-calculus with an additional operator $ P \setminus z $ called blocking. Blocking forbids for $ P $ to perform a visible action with the blocked name $ z $ as subject or bound object. By \cite{vivas} this allows to encode a match prefix $ \match{a}{b}P $ by the term $ \Res*{w}{\left( \piOutput{a}{y}{\nullTerm} \mid \piInput{b}{z}{\piOutput{w}{y}{\nullTerm}} \right) \setminus a \setminus b \mid \piInput{w}{z}{\left\lbr P \right\rbr}} $, where $ \left\lbr P \right\rbr $ is the encoding of $ P $ and $ z \notin \freeNames{P} $. As suggested by our proof above, the match prefix is translated into a communication and the match variables are translated into the channel names of the respective communication partners. To communicate the channel names have to be equal, \ie again either $ a = b $ holds from the beginning or $ a $ and $ b $ have to be unified by a substitution induced by a surrounding input. To enable such a substitution, the match variables $ a $ and $ b $ {have}|{as} shown in our proof {above}|{to} be translated into free names. Here the new blocking operator ensures that there are no unintended interactions between the translated match variables of parallel match encodings. More precisely $ M \setminus a \setminus b $ ensures that $ M $ cannot interact with another term over $ a $ or $ b $|{thus} blocking behaves as a binding operator \wrt reduction {steps}|{but} blocking does not bind the names $ a $ and $ b $ such that they can be affected by substitution.
Thus our proof explicitly reveals the features that due to \cite{vivas} allow to encode the match prefix by means of blocking.

\cite{carbone} extends the \piCal by so-called polyadic synchronisation, \ie instead of single names as in the \piCal channel names can be constructed by combining several names. Thus \eg in the variant of the \piCal with polyadic synchronisation, where each channel name consists of exactly two names, the input prefix becomes $ \piIn{x_1 \cdot x_2}{z} $ and the (matching) output prefix becomes $ \piOut{x_1 \cdot x_2}{y} $. An input and an output guarded term (that are composed in parallel) can communicate if the composed channel names are equal. By \cite{carbone} this extension allows to encode the match prefix. They suggest to translate $ \match{a}{b}P $ by $ \Res*{x}{\piOut{x \cdot b}{y} \mid \piInput{x \cdot a}{z}{\left\lbr P \right\rbr}} $, where $ \left\lbr P \right\rbr $ is the encoding of $ P $ and $ x, z \notin \freeNames{P} $. Again, as suggested by our proof above, the match prefix is translated into a communication and the match variables are translated into (parts of) the channel names of the respective communication partners. But polyadic synchronisation allows to combine the free match {variables}|{used} to allow for a guarding input to receive a {value}|{and} the bound name $ x $|{used} to avoid unintended interactions between the translated match variables of parallel match {encodings}|{within} a single communication channel.
Again our proof explicitly reveals the features that due to \cite{carbone} allow to encode the match prefix by means of polyadic synchronisation.

\section{Conclusions}
\label{sec:conclusions}

We provide a novel separation result showing that there is no valid encoding from the full \piCal into its variant without the match prefix. In contrast to the former approach in \cite{carbone} we strengthen the result in two ways:
\begin{compactenum}
	\item We considerably weaken the set of requirements, in particular with respect to the criterion that is called reasonable semantics in \cite{carbone}. Instead, we use the framework of criteria designed by Gorla for language comparison.
	\item The so obtained proof reflects our intuition on the match prefix and reveals the problem that prevents its encoding. A valid encoding of the match prefix would need to translate the prefix into a (set of) communication step(s) on links that result from the translation of the match variables. These links have to be free|to allow for a guarding input to receive a value for a match variable|but they also have to be bound|to avoid unintended interactions between parallel match encodings. This kind of binding cannot be simulated by a \piCal operator different from the match prefix.
\end{compactenum}
This further underpins that the match prefix cannot be derived in the \piCal.

In Section~\ref{sec:encodeMatchInOtherCalculi} we discuss four modifications and extensions of the \piCal that allow to encode the match prefix. 
In the first encoding approach the match prefix is replaced by another (more general) conditional guard. But the other approaches use extensions or modifications of the \piCal to encode the match prefix by using features that allow to circumvent the binding problem in the encoding of the match prefix that is pointed out in our proof.
Thus further works can use the here presented explicit formulation of the reason, that forbids for encodings of the match prefix in the \piCal, to encode the match prefix in other calculi.

\addcontentsline{toc}{section}{References}
\bibliographystyle{eptcs}
\bibliography{references}

\end{document}